\newif\ifFull
\Fulltrue
\ifdefined\EXTENDED
\Fullfalse
\else
\fi
\newif\ifExtended
\ifFull\Extendedfalse\else\Extendedtrue\fi

\documentclass[twoside,leqno,twocolumn]{article}

\usepackage[letterpaper]{geometry}

\usepackage{ltexpprt}

\usepackage{amsmath,amssymb,thmtools}
\usepackage{booktabs,color,doi,graphicx,latexsym,url,xcolor,xspace}
\usepackage{cleveref}
\usepackage{microtype,hyperref}
\usepackage{eucal}
\usepackage[inline]{enumitem}
\definecolor{citecolor}{HTML}{0000C0}
\definecolor{urlcolor}{HTML}{000080}
\newcommand{\RomanNumeralCaps}[1]
{\MakeUppercase{\romannumeral #1}}
\usepackage{tikz}
\usepackage{varwidth}
\usepackage{balance}
\usepackage{algorithm}
\usepackage[noend]{algorithmic}
\usepackage{multirow}
\usepackage{colortbl}
\usepackage{hhline}
\usepackage[multiple,para]{footmisc}
\usepackage[numbers,sort&compress]{natbib}

\newcommand{\G}{\ensuremath{\mathcal{G}}}
\newcommand{\cont}{\ensuremath{\text{cont}}}

\newtheorem{observation}{Observation}
\definecolor{lgray}{rgb}{0.95,0.95,0.95}
\hypersetup{
    colorlinks=true,
    linkcolor=black,
    citecolor=citecolor,
    filecolor=black,
    urlcolor=urlcolor,
}

\begin{document}

\title{\Large On the Feasibility of Perfect Resilience with Local Fast Failover\vspace{-3mm}\thanks{Supported by Vienna Science and Technology Fund (WWTF) project WHATIF, ICT19-045, 2020-24}}
\author{\large Klaus-Tycho\ Foerster$^\ddagger$ Juho Hirvonen$^\#$ Yvonne-Anne Pignolet$^\dagger$ Stefan Schmid$^\ddagger$ Gilles Tredan$^*$
\and \small{$^\ddagger$Faculty of Computer Science, University of Vienna, Austria~~~$^\#$Aalto University, Finland}
\and \small{$^\dagger$DFINITY, Switzerland~~~$^*$LAAS-CNRS, France}
\vspace{-4mm}
}

\date{}

\maketitle
\pagestyle{myheadings}
\markboth{On the Feasibility of Perfect Resilience with Local Fast Failover}{On the Feasibility of Perfect Resilience with Local Fast Failover}

\pagenumbering{arabic}

\begin{abstract} \small\baselineskip=9pt 

In order to provide a high resilience and to react quickly to
link failures, modern computer networks support fully decentralized
flow rerouting, also known as
local fast failover.
In a nutshell, the task of a local fast failover algorithm
is to \emph{pre-define}  fast failover rules
for each node using locally available information only.
These rules determine for each incoming link from which a packet may arrive
and the set of local link failures (i.e.,
the failed links \emph{incident} to a node),
on which outgoing link a packet should be forwarded.
Ideally, such a local fast failover algorithm
provides a \emph{perfect resilience} deterministically: a packet emitted from any source can reach
any target, as long as the underlying network
remains connected.
Feigenbaum et al.\ (ACM PODC 2012) and also Chiesa et al.\ (IEEE/ACM Trans.\ Netw.\ 2017) showed that it is not always possible to provide perfect resilience.
Interestingly, not much more is known currently about the
feasibility of perfect resilience.

This paper revisits perfect resilience
with local fast failover, both in a model where the source
can and cannot be used for forwarding decisions.
We first derive several fairly general impossibility results:
By establishing a connection between graph minors and resilience, 
we prove that it is impossible to achieve perfect resilience on \emph{any non-planar graph}; furthermore, while planarity is necessary, it is also not sufficient for perfect resilience. In some scenarios, a local failover algorithm cannot even guarantee that a packet reaches its target, even if the source is still highly connected to the target after the failures. 

On the positive side, we show that graph families closed under link subdivision allow for simple and efficient failover algorithms which simply  \emph{skip failed links}.  We demonstrate this technique by deriving perfect resilience for outerplanar graphs and related scenarios, as well as for scenarios where the source and target are topologically close after failures.

\vspace{-1mm}
\end{abstract}

\section{Introduction}
\label{sec:intro}
The dependability of distributed systems often critically
depends on the underlying network, realized by a set
of routers. To provide high availability, modern
routers support local fast rerouting of flows:
routers can be pre-configured
with conditional failover rules which define, for
each incoming port and desired target, to which
port a packet arriving on this incoming port
should be forwarded deterministically depending on the status of the
\emph{incident} links only: as routers need to react
quickly, they do not have time to learn about
remote failures.

\noindent This paper is motivated by the following
fundamental question
introduced by local fast rerouting mechanisms:

\emph{Is it possible to pre-define deterministic local failover rules
which guarantee that packets reach their target, as long as the underlying
network is connected?}

This desired property is known as
\emph{perfect resilience}.
The challenge of providing perfect resilience
hence lies in the decentralized nature of the problem
and the fact that routers only have local information
about failed links; achieving perfect resilience is
straightforward with global knowledge, as one could
simply compute a shortest path.

Unfortunately, perfect resilience cannot be achieved in general: Feigenbaum et al.~\cite{podc-ba,DBLP:journals/corr/abs-1207-3732} presented an example with 12 nodes for which, after certain failures, no forwarding pattern on the original network allows each surviving node in the target's connected component to reach the target.
Chiesa et al.~\cite{DBLP:journals/ton/ChiesaNMGMSS17} expanded on their result to require only two failures on a planar graph, but required over 30 nodes.
On the positive side, Feigenbaum et al.\ showed that it is at least always possible to tolerate one link failure, i.e., to be 1-resilient.
Interestingly, not much more is known today about when perfect resilience can be achieved, and when not.

\vspace{-3mm}

{
\begin{table*}[!t]
\setlength\extrarowheight{1pt}
\begin{center}
\begin{small}
\begin{tabular}{>{\centering\arraybackslash}p{2.7cm}>{\centering\arraybackslash}p{5.2cm}>{\centering\arraybackslash}p{5.2cm}}
  \textbf{Graph class}
&
  \textbf{Without source matching}
&
  \textbf{With source matching}
\\ \hline\hline
\rowcolor{lgray}
\arrayrulecolor{white}
Outerplanar & \cellcolor{green!15} Perfect resilience: Thm~\ref{thm:outerplanar} & \cellcolor{green!15} Perfect resilience (see left)\\ \hline
  $K_4$ &   \cellcolor{green!15} Perfect resilience: Thm~\ref{thm:k4} & \cellcolor{green!15} Perfect resilience (see left)\\\hline \rowcolor{lgray}
  Planar~graphs & \cellcolor{yellow!15} $|V|=7$ counterexample:
                  Thm~\ref{thm:planar} & \cellcolor{yellow!15} $|V|=8$ counterexample: Thm~\ref{thm:planar-s}\\\hline
	Non-planar graphs & \cellcolor{red!15} Perfect res. impossible: Thm~\ref{thm:planar-wagner} & ? \\
	\arrayrulecolor{black}
  \hline
\end{tabular}
\end{small}
\end{center}
\vspace{-6mm}
\caption{Summary of perfect resilience results for specific graph classes.}%
\vspace{-2.5mm}
\label{tbl:big1}

\setlength\extrarowheight{1pt}
\begin{center}
\begin{small}
\begin{tabular}{>{\centering\arraybackslash}p{2.7cm}>{\centering\arraybackslash}p{5.2cm}>{\centering\arraybackslash}p{5.2cm}}
  \textbf{Graph class}
&
  \textbf{Without source matching}
&
  \textbf{With source matching}
\\ \hline\hline
\arrayrulecolor{white}
\rowcolor{lgray}
  Closed under link subdivision & \multicolumn{2}{p{10.84cm}}{\cellcolor{green!15} \centering{Perfect resilience
                                $\Rightarrow$ skipping perfect resilience, \newline $f(m)$-resilience $\Rightarrow$ skipping $f(m)$-resilience:
                                Thm~\ref{thm:sim-argument}}} \\ \hline
\rowcolor{lgray}
   General graphs & \multicolumn{2}{c}{ \cellcolor{red!15} No superconstant resilience: Thm~\ref{thm:feigenbaum-padded}} \\ \hline 
  General graphs & \multicolumn{2}{c}{ \cellcolor{red!15} Impossible even if large connectivity under $F$: Thm~\ref{thm:feigenbaum-padded-replicated}} \\ \hline  \rowcolor{lgray}
Subgraph & \cellcolor{green!15} Perfect resilience\ kept: Thm~\ref{result:subset-s} & \cellcolor{green!15} Perfect resilience\ kept: Cor~\ref{result:subset-source}\\ \hline
Minor & \cellcolor{green!15} Perfect resilience\ kept: Thm~\ref{result:minor} & \cellcolor{green!15} Perfect resilience\ kept: Cor~\ref{minorstab-source} \\
\arrayrulecolor{black}
  \hline
\end{tabular}
\end{small}
\end{center}
\vspace{-6mm}
\caption{Characteristics of forwarding patterns and parametrized resilience results.}%
\vspace{-2.5mm}
\label{tbl:big2}

\setlength\extrarowheight{1pt}
\begin{center}
\begin{small}
\begin{tabular}{>{\centering\arraybackslash}p{2.7cm}>{\centering\arraybackslash}p{5.2cm}>{\centering\arraybackslash}p{5.2cm}}
  \textbf{Context}
&
  \textbf{Without source matching}
&
  \textbf{With source matching}
\\ \hline\hline
\arrayrulecolor{white}
\rowcolor{lgray}
$G$ planar and $s,t$ on the same face & \cellcolor{green!15} Perfect resilience: Cor~\ref{cor:sameface} & \cellcolor{green!15} Perfect resilience (see left)\\ \hline
Small distance in general graphs&  \cellcolor{green!15} All nodes at most two hops from $t$
                                  in $G\setminus F$: perfect resilience: Thm~\ref{thm:2hops-no-s} & \cellcolor{green!15} $s$ and $t$ at most two hops away in $G\setminus F$: perfect resilience: Thm~\ref{thm:2hops}\\
\arrayrulecolor{black}
\hline
\end{tabular}
\end{small}
\end{center}
\vspace{-6mm}
\caption{Perfect resilience under specific settings. 
}%
\label{tbl:big3}
\vspace{-4mm}
\end{table*}
}

\subsection{Contributions}

This paper studies the problem of providing perfect resilience in networks where nodes only have local information, considering both a model where nodes can and cannot match the packet source. 

On the negative side, we 
%
show that perfect resilience is impossible already on simple and small planar graphs, and even in scenarios where a source is in principle still highly connected to the target after the failures, by $\Omega{(n)}$ disjoint paths; however, it cannot route to it.
We also derive the general negative result that perfect resilience is impossible on any non-planar graph.
To this end, we show an intriguing connection to graph minors, 
where every graph minor retains the perfect resilience property.

On the positive side, we describe perfectly resilient algorithms for all outerplanar graphs and related scenarios (e.g., scenarios where the source and the destination are on the same face after failures or where the graph after removal of the destination is outerplanar), 
as well as for non-outerplanar scenarios where the destination is within two hops of the source. 
%
For our positive results, we establish the general insight that  graph families that are closed under subdivision of links, allow for simple failover algorithms in which nodes can just skip locally failed ports, requiring very small forwarding tables. 

Our results are summarized in the Tables~\ref{tbl:big1} to \ref{tbl:big3}. 


\subsection{Related Work}

Besides the paper by Feigenbaum et al.~\cite{podc-ba,DBLP:journals/corr/abs-1207-3732}, which is the closest work to ours, and the paper by Chiesa et al.~\cite{DBLP:journals/ton/ChiesaNMGMSS17}, the design of local fast failover algorithms has already been studied intensively, see the recent survey~\cite{frr-survey}.
In the following, we will concentrate on related work that provides formal
 resilience guarantees for local fast failover, however, we point out
that there also exists much interesting applied
work on the topic, e.g.,~\cite{keep-fwd,purr,ref7,ref4,ref34,srds19failover,plinko-full,schapira1,schapira2}.

A key property and challenge in the design of local fast failover
algorithms is
related to the fact that nodes need to react to failures fast,
i.e., routers can only have local failure information and
failover decisions are static, ruling out algorithms
based on link reversal~\cite{gafni-lr,corson1995distributed}
or reconvergence~\cite{DBLP:conf/spaa/BuschST03}.
Furthermore, it is not possible to rewrite packet headers~\cite{elhourani2014ip,infocom15},
e.g., to carry failure information, which is often impractical;
this rules out graph exploration algorithms such as~\cite{hotsdn14failover,reingold,DBLP:journals/tcs/FoersterW16,DBLP:journals/tcs/MegowMS12}.
The fact that neither routing tables
nor packet headers can be modified also means that it is
not possible to adopt, e.g., rotor router approaches~\cite{rotorrouter}, which would provide connectivity but require state.
We also note that the focus in this paper is on \emph{deterministic} algorithms,
which, in contrast to related work such as \cite{icalp16,frr-rand}, 
do not require random number generators at routers.

From a distributed perspective, the design of local fast failover algorithms also features an interesting connection to distributed algorithms ``with disconnected cooperation''~\cite{malewicz2006distributed}, a subfield of distributed computing where nodes solve a problem in parallel without exchanging information among them.
This was first pointed out by Pignolet et al.~in~\cite{dsn17}, and used in several related works~\cite{opodis13shoot,infocom19casa}, however, primarily to reduce load in dense networks, rather than to improve resilience.

In terms of local fast failover algorithms which provide a provably high resilience, there exist several interesting results by Chiesa et al.\ who introduced a powerful approach which decomposes the network into arc-disjoint arborescence covers~\cite{icalp16,robroute16infocom,DBLP:journals/ton/ChiesaNMGMSS17}, further investigated in~\cite{infocom19casa,dsn19,srds19failover,ccr18failover} to reduce stretch and load.
Chiesa et al.\ are primarily interested in $k$-connected graphs, and leverage a well-known result from graph
theory, which allows us to decompose any $k$-connected graph into a set of $k$ directed spanning trees (rooted at the same vertex, the target) such that no pair of spanning trees shares a link in the same direction.
Packets are routed along some arborescence until hitting a failure, and are then rerouted along a different arborescence.
Chiesa et al.\ conjecture that for any $k$-connected graph, basic failover routing can be resilient to any $k-1$ failures, and show that this conjecture is true for a number of interesting graphs, including planar graphs.

However, the general conjecture
remains to be proved, and also differs conceptually
from the perfect resilience considered in our paper, which
asks for connectivity as long as the underlying graph is connected.
Indeed, as we show,
%
perfect resilience is
for example \emph{not}
achievable in planar graphs, even with the packet source.

%

\vspace{4mm}

\noindent\textbf{Bibliographical Note} This work also appears in the APOCS21 conference proceedings~\cite{DBLP:journals/corr/abs-2006-06513} and a three page brief summary of the main results is available at~\cite{disc-bam}\footnote{\url{https://www.youtube.com/watch?v=m3xWggfbHo4}}.

\section{Model}
\label{sec:model}

Let $G=(V,E)$ be a network represented by an undirected graph of nodes (``routers'') 
$V$ connected through undirected links $E$ along which packets are exchanged.
Initially,
an arbitrary set $F\subset E$ of links fail (rendering them unusable in both directions). We write $G \setminus F$ to denote the graph $G$ with the links in $F$ removed. More generally, $G \setminus E'$ and $G \setminus V'$ represent the graph $G$ after removing the set of links in $E' \subset E$ or the set of nodes $V' \subset V$ and their incident links respectively.
We define $n=|V|$, $m=|E|$, and write $V_G(v)$ and $E_G(v)$ for the neighbors and incident links of node $v$, respectively. $G$ is omitted and the  degree $d_v=|V_G(v)|$ is shorthanded to $d$ when the context is clear.
We study the class of local routing (forwarding) algorithms, in which every node
$v\in V$ takes deterministic routing decisions on
\begin{itemize}
\setlength\itemsep{0em}
\item the target $t$ of the packet to route,
\item the set of incident failed links $F\cap E(v)$, and
\item the receiving or incoming port (\emph{in-port}) of the packet at node $i$.\footnote{Note that without knowledge of the in-port, already very simple failure scenarios prevent resilience. For example, consider a packet reaching a node $v$ from a node $w$, where all further links incident to $v$ failed. Node $v$ must return the packet to $w$, which forwards it back to $v$, resulting in a permanent forwarding~loop.}
\end{itemize}

This implies that neither the state of the packet nor the state of the node can be changed, e.g., by header rewriting
or using dynamic routing tables.
As such, routing has be to purely local, which we formalize next.

Given a graph $G$ and a target $t \in V(G)$, a local routing algorithm is modelled as a \emph{forwarding function} $\pi_v^t \colon E(v) \cup \{ \bot \} \times 2^{E(v)} \mapsto E(v)$ at each node $v \in V(G)$, where $\bot$ represents the empty in-port, i.e.\ the starting node of the packet. That is, given the set of failed links $F \cap E(v)$ incident to a node $v$, a forwarding function $\pi_v^t$ maps each incoming port (link) $e=(u,v)$ to the outgoing port (link) denoted by $\pi_v^t(e,F)$ or $\pi_v^t(u,F)$.
The tuple of forwarding functions $\pi^t = (\pi_v^t)_{v \in V}$ is called the \emph{forwarding pattern}.
When talking about the repeated use of a forwarding function of the node $v$ under a specific failure set $F$, we will also use the notation style~$\pi^t_v(\cdot, F)$, where for $a \in \mathbb{N}_{\geq 1}$, $\left(\pi^t_v(\cdot, F)\right)^a(u)$ is recursively defined as $\left(\pi^t_v(\cdot, F)\right)^{(a-1)}\left(\pi^t_v(u, F)\right) =\left(\pi^t_v(\cdot, F)\right)^{(a-2)}\left(\left(\pi^t_v(\cdot, F)\right)^2(u)\right)$ etc.


We say that a forwarding pattern is \emph{$k$-resilient} if for all $G$ and all $F$ where $|F| \leq k$ the forwarding pattern routes the packet from all $v \in V$ to the target $t$ when $v$ and $t$ are connected in $G \setminus F$. A forwarding pattern is \emph{perfectly resilient} if it is $\infty$-resilient: the forwarding always succeeds in the connected component of the target.
Let $A_p(G,t)$ be the set of such perfectly resilient patterns (algorithms), abbreviated by $A_p$ when the context is~clear.

We also consider the model where the forwarding functions depend on the source\footnote{Note that matching on the source greatly increases routing table size, from one forwarding pattern for every destination, to one forwarding pattern for every source-destination pair.} node~$s$ in addition to the in-port, incident failures and target.
We denote the forwarding function and pattern by $\pi_v^{s,t}$ and $\pi^{s,t}$, respectively; it is perfectly resilient if forwarding always succeeds when $s$ and $t$ are connected in $G \setminus F$, where the set of such patterns is denoted by $A_p(G,s,t)$.

Note that resilient algorithms that do not match on the source also imply
resilient algorithms in the model where we can match on the source (it can be ignored); similarly, impossibility results in the model where we can match on the source also imply impossibility in the model where we do not match on the source (extra~information).

\section{First Insights}\label{sec:first}

Let us start by providing some basic insights and also derive a first impossibility result.
In general, we observe that to achieve perfect routing, every possible route to the target must be explored.
The challenge is that 
the nodes must ensure this property in a distributed fashion, while having no means to exchange their local views of the network.
%
%
%
We will show that in order to achieve this, the forwarding functions must be coordinated for global coherence.
We will then study the limits of such patterns.
But first, we formalize the ``local'' perspective of a node, illustrated in Figure~\ref{fig:relevance}.

	\begin{figure}[t]
  \begin{center}

  \begin{tikzpicture}[shorten >=1pt]
  \tikzstyle{vertex}=[circle,fill=black!25,minimum size=17pt,inner sep=0pt]
  \tikzstyle{edge}=[thick,black,--]

  \node[vertex,fill=green!25] (i) at (0,0) {$i$};


  \node[vertex] (v1) at (120:1.3cm) {$v_1$};
  \node[vertex] (v2) at (45:1.3cm) {$v_2$};
  \node[vertex] (v3) at (0:1.3cm) {$v_3$};

  \node[vertex,fill=blue!25] (source) at (-1,-1) {$s$};
  \node[vertex,fill=blue!25] (target) at (3,0) {$t$};

  \foreach \from/\to in {i/v1,i/v2,i/v3,target/v2,target/v3,v1/v3,v1/v2,i/source}
  { \draw (\from) -- (\to); }

\end{tikzpicture}
  \end{center}
	\vspace{-5mm}
\caption{In this context, only nodes $v_2$ and $v_3$ are relevant for   $i$: it is not useful for $i$ to relay to $v_1$ as $v_2$ and $v_3$
  constitute more direct alternatives for the same paths towards $t$. Note that relevance is lower bounded by connectivity in the failure-free case and that the number of relevant neighbors can increase with failures, e.g., if $(i,v_2)$ or $(i,v_3)$ fail, $v_1$ becomes relevant for $i$}
  \label{fig:relevance}
\end{figure}
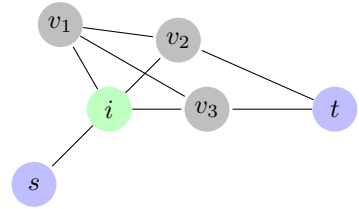

\begin{Definition}[Relevant node]
  For any node $i\in G$,
	$i \neq t$,
	and a failure set $F$, define $F_i$ as the failures in $F$ incident to $i$, 
	i.e., $F_i$ is the only failure set the node $i$ is aware of.
	Moreover, let $G'$ be the original graph $G$ without the links in $F_i$, i.e., $G'=G\setminus F_i$.
	
	A neighboring node $j \in V_{G'}(i)$ is 
		\emph{relevant} for routing to $t$ under the failure set $F$
	iff there is a path from $i$ to $t$ in $G' \setminus V'$, where $V'=V_{G'}(i)\setminus \{j\}$ is the set of all other nodes still connected to $i$.
	In other words, $j$ is a potential relay to reach $t$ from $i$'s perspective, if, in addition to $F$, all links incident to other neighbors of $i$ have failed.
  \end{Definition}

Note that nodes are oblivious to non-incident link failures and relevance is defined by a path to the target which does not use further neighbors of the current node.
Hence, if a neighbor $v$ of some node $i$ is relevant, then the failure of further links incident to $i$ can only remove $i$'s relevance if the link $(v,i)$ fails.\footnote{Note that a node itself does not have a global view and can only judge relevance based on its incident failures.}
\begin{observation}[Monotonicity of relevance]\label{obs:relevance}
If a node $v$ is relevant for a node $i$ under a failure set $F$, $v$ remains relevant for $i$ under any failure superset $F'$ with $F \subseteq F'$, if $v$ remains a neighbor of $i$ under $F'$.
\end{observation}

We next investigate forwarding under perfect resiliency when a packet arrives from a relevant neighbor.
As each relevant neighbor might be the only way of reaching the destination, the forwarding must permit the packet to ``try'' all relevant neighbors, no matter from which one the packet comes first, i.e., the only way of doing so is in a circular fashion.
%

\begin{Definition}[Orbit]
Let $\pi^t_v(\cdot, F)$ be the forwarding function of a node $v$ for some set of failed links $F$.
We say a set of neighbors $V' \subseteq V(v)$ is in the same orbit w.r.t.\ $\pi^t_v(\cdot, F)$, if for all pairs $v_1,v_2 \in V'$ it holds: there is some $k \in \mathbb{N}$ s.t.\ $\left(\pi^t_v(\cdot, F)\right)^k(v_1) = v_2$.
\end{Definition}

\begin{lemma}\label{lemma:orbit}
  Let $G,A\in A_p(G,t)$. For all $F$ and all $t,i \in V$ it holds: If $i$ has at least two relevant neighbors under $F$, then all relevant neighbors of $i$ must be part of the same orbit in $A$'s forwarding function $\pi^t_i(\cdot, F)$ of $i$.
\end{lemma}

\begin{proof}
Our proof will be by contradiction.
Assume there is a failure set $F$ and a node $i \in V$ s.t.\ there are relevant neighbors $v_1,v_2$ of $i$ where there is no $k \in \mathbb{N}$ s.t.\ $\left(\pi^t_v(\cdot, F)\right)^k(v_1) = v_2$.
In other words, $v_1,v_2$ are relevant for $i$, but not in the same orbit of $\pi^t_v$ for $F$.

Assume the packet starts on $v_1$, and beyond $F$, we fail all further links, except for those still incident to $i$ and links forming a path from $v_2$ to $t$ without visiting any other neighbors of $i$ except $v_2$.
Such a path must exist as $v_2$ is relevant for $i$.
The only routing choice at $v_1$ now is to route to $i$.  Observe that routing under $\pi_{i}$ will never reach $v_2$, as each packet leaving $i$ will immediately bounce back to $i$ (except if it would go to $v_2$).
Hence, one will not reach $t$, even though there is a path from $i$ to $t$ under $F$, and thus $A$ is not perfectly~resilient.
\end{proof}

We now extend Lemma~\ref{lemma:orbit} to allow matching on the source node as well.
However, we can now no longer enforce that the packet starts on an arbitrary neighbor of the node $i$, and hence obtain slightly modified results.
Essentially, we need to enforce that the packet can reach node $i$ from $s$, and that this taken path is node-disjoint from other paths from $i$ to $t$.

\begin{lemma}\label{lem:s-orbit}
  Let $G=(V,E),A\in A_p(G,s,t)$, where $s \neq i$ is adjacent to two different relevant neighbors $v_1,v_2\in V$ of $i \in V$. For all $F$ where $F \cap \left\{(v_1,i),(v_2,i)\right\} = \emptyset$, i.e., both $v_1,v_2$ are still neighboring $i$, it holds that all relevant neighbors of $i$ under $F$ must be part of the same orbit in $A$'s forwarding function~$\pi^{s,t}_i(\cdot,F)$.
\end{lemma}

\begin{proof}
We use an analogous argumentation as for the proof of Lemma~\ref{lemma:orbit}, where we will fail links not incident to $i$, s.t.\ $i$ cannot further adapt its forwarding function.
We start with the case where we fail all incident links of $v_1$ except $(v_1,i)$ and $(v_1,s)$.
Next, pick any further relevant neighbor $j \neq v_1$ of $i$, and fail all links in the graph except $1)$ a path from $j$ to $t$, not using further nodes from $V(i)$ (due to $j$ being relevant), $2)$ $(v_1,i),(v_1,s)$, and $3)$ all links still alive incident to $i$.
Then there must be some $k\in \mathbb{N}$ s.t.\ $\left(\pi^{s,t}_i(v_1, F)\right)^k = j$.
The same holds for $v_2$, i.e., every further relevant neighbor of $i$ will be reached by iterating the forwarding function, as that neighbor might be the only connection to the target $t$.
As every relevant neighbor of $i$ (including $v_2$) is in this way reachable from $v_1$, and as every relevant neighbor of $i$ (including $v_1$) is in this way reachable from $v_2$, all relevant neighbors of $i$ are in the same orbit w.r.t.\ $\pi^{s,t}_i(\cdot, F)$.
\end{proof}

Note that the above proof can be directly extended to $s$ being connected to multiple relevant neighbors of $i$, as long as two such neighbors remain non-faulty:
\begin{corollary}\label{corr:s-orbit-ext}
  Let $G=(V,E),A\in A_p(G,s,t)$, where $s \neq  i$ is connected to $k\geq 2$ relevant neighbors $v_1,\ldots,v_k\in V$ of $i \in V$. For all $F$ where $\left| F \cap \left\{(v_1,i),\ldots,(v_k,i)\right\}\right| \leq k-2$ it holds that all relevant neighbors of $i$ under $F$ must be part of the same orbit in $A$'s forwarding function $\pi^{s,t}_i(\cdot,F)$.
\end{corollary}

Observe that the previous proof arguments rely on the path from the source $s$ to $i$ via some relevant neighbor $v_1$ being node-disjoint from the path from some other relevant neighbor $v_2$ of $i$ to the target $t$.
If the above statement also holds vice-versa, $v_1,v_2$ need to be in the same orbit.
We cast this insight into the following corollary:

	\begin{corollary}\label{corr:s-paths-deg}
	 Let $G=(V,E),A\in A_p(G,s,t)$ with $i \neq s$. Fix any $F$ where $i$ has $k$ relevant neighbors, $k' \geq 2$, denoted $v_1,\ldots,v_{k'}$. In this case there is a path from $s$ to $i$ via $v_x$ which is node-disjoint to a path from $i$ to $t$ via some $v_y$, $y \neq x, 1 \leq y \leq k'$ for each individual $v_x$, $1 \leq x \leq k'$.
	Then, all $v_1,\ldots,v_{k'}$ must be part of the same orbit in $A$'s forwarding function $\pi^{s,t}_i(\cdot,F)$ of $i$.
	\end{corollary}

We now analyze our results in the context of prior work, which considered forwarding without the source.
\smallskip

\noindent\textbf{Relation to the work of Feigenbaum et al.~\cite{DBLP:journals/corr/abs-1207-3732}.}
Note that Lemma~\ref{lemma:orbit} implies corollaries reminiscent to two results presented by Feigenbaum et al.~\cite[Lemma 4.1, 4.2]{DBLP:journals/corr/abs-1207-3732} namely that
\begin{itemize}\setlength\itemsep{0em}
	\item if all neighbors are relevant, the forwarding function 
	must be a cyclic permutation, cf.~\cite[Lemma 4.2]{DBLP:journals/corr/abs-1207-3732},
	\item if a node has only two neighbors $v_1,v_2$ and both are relevant, then a packet from $v_1$ must be forwarded to $v_2$ and vice versa, cf.~\cite[Lemma 4.1]{DBLP:journals/corr/abs-1207-3732}.
\end{itemize}
Feigenbaum et al.\ use these lemmas to prove the impossibility of perfect resilience for a specific graph.
Their lemmas implicitly define a version of relevance based on link-connectivity, whereas we are interested in a node version, namely that a neighbor $v$ of $i$ is only relevant (from the viewpoint $i$) if there is a routing from $v$ to $t$ that does not use any other neighbor of $i$.
As such, their assumptions are weaker and hence their results are stronger.
However, both Lemmas 4.1 and 4.2 from Feigenbaum et al.~\cite{DBLP:journals/corr/abs-1207-3732} do not apply to all graphs, beyond the construction in their paper.
For a counterexample, consider the graph in Figure~\ref{fig:counter}.
If $u$ sends the packet to $v$, then both Lemma 4.1, 4.2 from Feigenbaum et al.\ state that the packet \emph{must not} be forwarded back to $u$, i.e., for perfect resilience, it \emph{must} next be sent to $w$, and from there to $u$.
Yet, $v$ could also forward the packet back to $u$, and then to $x$, a contradiction to their claim: the only two possible neighbors of $t$ are $x$ and $v$, hence visiting both~suffices.
%

\begin{figure}[b]
  \begin{center}

  \begin{tikzpicture}[shorten >=1pt]
  \tikzstyle{vertex}=[circle,fill=black!25,minimum size=17pt,inner sep=0pt]
  \tikzstyle{edge}=[thick,black,--]

	\node[vertex] (t) at (0,0) {$t$};
	\node[vertex] (x) at (2,0) {$x$};
	\node[vertex] (u) at (4,0) {$u$};
	\node[vertex] (v) at (6,0) {$v$};
	\node[vertex] (w) at (7.9,0) {$w$};

	\draw (t) to (x);
	\draw (x) to (u);
	\draw (u) to (v);
	\draw (v) to (w);

	\draw[dashed,red] (t) to[out=340,in=200] (v);
	\draw (u) to[out=20,in=160] (w);

\end{tikzpicture}
  \end{center}
\caption{Graph with target $t$ where the dashed red link has failed. Assume the packet starts at $u$ and is forwarded to $v$. The node $v$ can bounce the packet back.
}
	\label{fig:counter}
\end{figure}
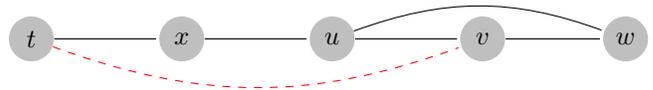

\section{A Minor Perspective on Perfect Resilience}
\label{sec:minor}

We now study the relationship between graph minors and perfect resilience algorithms. As we will see, this relationship will be relevant for the design of local resilient algorithms as well. We will first observe that a perfectly resilient algorithm is also perfectly resilient on subgraphs and contractions of its original graph (Theorems~\ref{result:subset-s} and~\ref{cont2}).
Then, since subsetting and contracting are the two fundamental operations in the minor relationship, we deduce that the existence of a perfectly resilient algorithm on a graph $G$ implies its existence on any minor of $G$ (Theorem~\ref{result:minor}).

As an application of the above results, we prove in \S\ref{subsec:nonplanar} that if a graph is not planar, it does not allow for perfect resilience (Theorem~\ref{thm:planar-wagner}), as both $K_5$ and $K_{3,3}$ do not allow for perfect resilience (Lemmas~\ref{nok5} and~\ref{nok3}).

\subsection{Contraction and Minors}
We first recall that a graph $G'$ is a minor of the graph $G$ if we can obtain $G'$ from $G$ by means of contraction (see Definition~\ref{def:contraction}) and by removing links and nodes.
 We hence start by retaining perfect resilience on subgraphs:

\begin{theorem}[Subset Stability]\label{result:subset-s}
\label{subset}
If a perfectly resilient forwarding pattern exists for a graph $G$ and target $t$ then there also exists a perfectly resilient pattern for all subgraphs of $G$ when they contain $t$,
  $A_p(G,t) \neq\emptyset \rightarrow 
	G' \subset G, t \in V(G'), A_p(G',t) \neq \emptyset$
\end{theorem}

The proof idea is as follows: we take the original algorithm $A$, and fail all links needed to obtain the desired subgraph $G'$.
Then, we can simulate a perfectly resilient algorithm on $G'$ via $A$, where $A$ also needs to function correctly when further links fail in  subgraph~$G'$.

\begin{proof}
	Let $A\in A_p(G,t)$. Let $\bar{G'}= G\setminus G'$.
	Let $A'$ be a fast rerouting algorithm on $G'$ s.t.\  $\forall F'\subset E(G'), A'(F')=A(F' \cup E(\bar{G'}))$: algorithm   $A'$ simply executes $A$ in a context where all the parts of $G$  that are not in $G'$, i.e., in $\bar{G'}$ have also failed.
  Since $A \in A(G,t)$, we have that if $v,t$ are connected in $G\setminus(F'  \cup E(\bar{G'}))$ then $A$ succeeds in routing.
  Since $\forall v,t \in  V(G')$, if $v,t$ are connected in $G'\setminus F'$ then $(v,t)$  is connected in $G\setminus (F' \cup E(\bar{G'} ))$, and we deduce that $A$ achieves a perfect resilience on $G'$.
\end{proof}

Note that the above proof arguments hold analogously 
when we can match on the source $s$, where perfect resilience only needs to hold when $s,t$ remain in the same connected component.

\begin{corollary}[$s$ Subset Stability]\label{result:subset-source}
\label{subset-source}
  $A_p(G,s,t) \neq\emptyset \rightarrow 	G' \subset G, s,t \in V(G'), A_p(G',s,t) \neq \emptyset$
\end{corollary}

After covering subgraphs, we next cover node contractions.
Abstractly, a node contraction merges two neighboring nodes, while retaining their joint connectivity to the graph.

\begin{Definition}[Node Contraction]\label{def:contraction}
  Let $G=(V,E)$ and $i,j \in V$ be two neighboring nodes, 
	i.e., $(i,j)\in E(G)$.
	Let $G'$ be the $(i,j)$-contracted graph of $G$ s.t.:
  $V(G')=V(G)\setminus \{j\}$ and $E(G')=E(G)\cup \{(a,i), \forall a \in V_G(j)
  \} \setminus \{(a,j), \forall a\in V_G(j) \}$.
  We denote by $\cont(G)$ the set of all possible contracted graphs of $G$.
\end{Definition}

Note that one can define the forwarding pattern also independent of the destination (and/or source) as a port mapping, where a packet arriving at an in-port gets forwarded to some out-port, and hence might talk simply about forwarding patterns $\pi$ if the context is clear.
We next define the natural emulation of a
$(i,j)$-contraction with respect to such port mappings:

        \begin{Definition}[Mapping Contraction]
          Let $i,j$ be two neighboring nodes without common neighbors, $V(i)\cap V(j)= \emptyset$, and let
          $\pi_i$ and $\pi_j$ be their
          forwarding patterns. We define $\pi_{ij}:
          V(i)\cup V(j) \cup \{\bot\}\mapsto V(i)\cup V(j)$ 
					as the \emph{contracted}
          mapping where on $i$'s side $\forall v\in V(i)$, $\pi_{ij}(v)=\pi_i(v)$ unless
          $\pi_i(v)=j$; in the latter case $\pi_{ij}(v)=\pi_j(\pi_i(v))=\pi_{j}(i)$, unless
          $\pi_j(i)=i$ in which case  $\pi_{ij}(v)=\pi_i(j)$, unless
          $\pi_i(j)=j$; in this last case, $\pi_{ij}(v)=\perp$. 
          For $j$'s
          side, proceed symmetrically accordingly. 
					By convention, we set $\pi_{ij}(\perp)=\pi_i(\perp)$
        \end{Definition}

Note that for the above two definitions, matching on the source is immaterial, and hence the definitions are identical for forwarding patterns with and without the source.
For the following Theorem~\ref{cont2}, we prove that they preserve resilience when not matching on the source. 
In the proof of the result we show that the routing of a packet, starting at some node, towards the target $t$, can be directly transferred to the contracted case, by simulating the original algorithm.
For this simulation, it is again immaterial if we match on the source or not, described in Observation~\ref{obs:stab}.

\begin{theorem}[Contraction Stability]\label{cont2}
  Let $G=(V,E), i,j \in V$ be two neighboring nodes, and
	let $G'$ be the  $(i,j)$-contraction of $G$. 
  Given $A \in A_p(G,t)$, a perfectly resilient algorithm on $G$,
  then the $(i,j)$-contracted algorithm of $A$ is perfectly resilient on $G'$.
\end{theorem}

\ifExtended
The full proof is deferred to \cref{sec:cont1}.
\fi
\ifFull
 We will utilize the following lemma for the proof: 

\begin{lemma}[Algorithm Transfer]
\label{cont1}
Given $G$ and $(i,j)\in E(G)$ let $G'$ be the corresponding $(i,j)$-contraction.
	Let $R=\{(j,r),r\in  V(i)\cap V(j)\}$.
	Let $A:F\mapsto \{\pi^t_v(\cdot,F),v\in V\}$.
	Define $A': F\mapsto  \{\pi^t_v(\cdot,F),v\in V'\}$ to be the \emph{$(i,j)$-contracted algorithm} of $A$ as follows:
  \begin{itemize}
  \item Case \RomanNumeralCaps{1}: Identical behavior on unaffected nodes, $\forall v\in V', v\neq \{i\}$, $\pi'^t_v(\cdot,F)=\pi^t_v(\cdot,F\cup R)$.
  \item Case \RomanNumeralCaps{2}: Replace $i$'s forwarding by the contracted algorithm. Let $\pi'^t_i(\cdot,F)=\pi^t_{\{i,j\}}(\cdot,F\cup R)$.
  \item Case \RomanNumeralCaps{3}: Replace $j$'s port by $i$'s port on $j$'s neighbors
    forwarding: $\forall k\in V(j),\exists v \text{~s.t.~} \pi^t_k(v,F \cup R)=j
    \Rightarrow \pi'^t_k(v,F)=i$,  $\forall k\in V(j),\exists v \text{~s.t.~} \pi_k^t(j,F \cup R)=v
    \Rightarrow \pi'^t_k(j,F)=v$.
  \end{itemize}
Let $P$ (respectively $P'$) be the sequence of links traversed using $A$ by a
packet from $s$ to $t$ under $F\cup R$ (respectively, traversed using $A$'
under $F$). 
Let $Q$ be a rewriting of $P$ in which we replace every occurence of $j$
by $i$. And let $Q'$ be the rewriting of $Q$ in which we remove every
occurrence of $(i,i)$. We have $Q'=P'$.
\end{lemma}

\begin{proof}[Lemma~\ref{cont1}]
  We proceed by induction on the $k$ first hops  $Q[1..k]$ and
  $P'[1..k]$ of the sequences. When the context is clear, we write
  $\pi_v$ for the forwarding function of node $v$ using $A$ in context $F\cup R$,
  and $\pi'_v$ the forwarding function of $v$ using $A'$ in context $F$.

\emph{  Base case}: let $Q[1]=(s,q),P'[1]=(s,a)$. We need to prove
  $q=a$. $1)$: if $s\neq i$, $\pi'_s(\bot)=\pi_s(\bot)$ as defined in case
  \RomanNumeralCaps{1}. $2)$: if $s=i$, the packet starts in the
  $i,j$ contraction. Since $a=\pi'_s(\bot)=\pi_{i,j}(\bot)=\pi_i(\bot)$
  (Case \RomanNumeralCaps{2}),
  we deduce that $P$ describes possibly some hops between $i$ and $j$, and
  then $q$. Since $(i,j)$ transitions are rewritten $(i,i)$ in $Q$ and
  removed in $R$, the first node that is not $i$ nor $j$ to appear in
  $P$ must be $a$. Thus $q=a$.

\emph{Induction step:} Assume the following statement holds for any
$k'\leq k$: $Q[k']=P'[k']$. We prove that necessarily $Q[k]=P'[k]$.
Let $Q[k-1]=(v_-,v),P'[k-1]=(v_-,v),Q[k]=(v,q), P'[k]=(v,a)$. We need to
show that $a=q$. We again start by the simplest case $1)$: $v\not\in V'(i)
\cup \{i\}$: since $\pi_v=\pi'_v$ by Case \RomanNumeralCaps{1},
$q=\pi_v(v_-)=\pi'_v(v_-)=a$.
$2)$: if $v\in V'(i)$. We need to look at $v_-$. If $2.1)$: $v_-\neq i$ we
again directly use Case \RomanNumeralCaps{1}:
$\pi_v(v_-)=\pi'_v(v_-)$. If $2.2)$: $v_-= i$ the packet just left the
$(i,j)$ contraction. In $P$, $v_-$ can correspond to either $i$ or $j$
in $P$. If it corresponds to $i$, use Case \RomanNumeralCaps{1}. If it
was a $j$, since in Case \RomanNumeralCaps{3} we replaced $j$'s
connections (in $A$) by $i'$s connections (in $A'$), and since no node
$v$ is both a neighbor of $i$ and $j$ in $F\cup R$, this mapping
uniquely applies.
$3)$: if $v=i$: the packet is in the $(i,j)$
contraction. By definition of $Q$, necessarily $v_-\neq i$ or
$j$. Since $\pi'_i(v_-)=\pi_{ij}(v_-)$ we deduce $q=a$.
\end{proof}

\begin{proof}[Theorem~\ref{cont2}]
	Let $R=\{(j,r),r\in V(i)\cap V(j)\}$. 
	Let $A'$ be the $(i,j)$ contracted algorithm of $A$.
  Let $F$ be a set of link failures of $G'$. 
	Observe that if $s,t$ are connected in $G'\setminus F$
	then $s,t$ are connected in $G\setminus F\cup R$.
  Let $s,t$ be connected in $G'\setminus F$.
	Let $p_{A'}$ (resp. $p_{A}$) be the path produced by $A'$ on $G'\setminus F$ (resp.\ $A$ on $G\setminus (F \cup R)$) from $s$ to $t$.
  As $A\in A_p$ and $s,t$ are connected in $G\setminus (F\cup R)$,
  then $p_A$ is finite and ends up in $t$. Therefore,
  Lemma~\ref{cont1} implies $p_{A'}$ is finite and ends in $t$
  too: $A'$ succeeds.
\end{proof}
\fi

\begin{observation}[$s$ Contraction Stability]\label{obs:stab}
The results of Theorem~\ref{cont2} also apply when matching on the source, i.e., in particular a perfectly resilient algorithm $A\in A_p(G,s,t)$) implies
a perfectly resilient algorithm on a contracted graph $G'$, i.e., $A_p(G',s,t)\neq\emptyset$).
\end{observation}

We next combine Theorem~\ref{result:subset-s} (subsetting) and Theorem~\ref{cont2} (contraction) to obtain the corresponding result for minors.
In other words, if $G$ permits a   perfectly resilient scheme, so do its minors:

\begin{theorem}[Minor Stability]\label{result:minor}
\label{minorstab}
Given $G, G'$,  $G'$ a minor of $G$, it holds that %
$A_p(G,t) \neq \emptyset \Rightarrow A_p(G',t)\neq \emptyset$. 
\end{theorem}
\begin{proof}
	Observe that there exists a sequence of graphs $G_1,G_2\ldots$ s.t.\ $G\mapsto G_1\mapsto G_2   \mapsto \ldots\mapsto G'$, where $\mapsto$ is either a subsetting or a contraction operation.
	Since both operations preserve the existence of a perfect scheme thanks to Theorems~\ref{result:subset-s} and~\ref{cont2}, in
	combination,
	they imply the general minor relationship.
\end{proof}

The above proof transfers to the model with matching on the source, by utilizing 
Corollary~\ref{result:subset-source} and Observation~\ref{obs:stab}, instead of Theorems~\ref{result:subset-s} and~\ref{cont2}:

\begin{corollary}[Minor Stability]
\label{minorstab-source}
Given two graphs $G, G'$, $G'$ a minor of $G$, it holds that
 $A_p(G,s,t) \neq \emptyset$ implies that $A_p(G',s,t)\neq \emptyset$: if $G$ permits a
 perfectly resilient scheme, so do its minors.
\end{corollary}

The subset and contraction lemmas provide constructive proofs, showing how to derive a perfectly resilient scheme from a larger graph to one of its minors. We can also exploit this result in its contrapositive form: 
by showing the absence of perfect resilience schemes on the minors defining a minor-closed graph family, we can prove the impossibility of perfectly resilient schemes for whole graph families, e.g., planar graphs.

\subsection{Case Study: Non-Planar Graphs}\label{subsec:nonplanar}

In order to prove the impossibility of perfect resilience on non-planar graphs, we first cover the impossibility of perfect resilience on the $K_5$ (the complete graph with five nodes) and then on the $K_{3,3}$ (the complete bipartite graph with three nodes in each partition).
Note that our impossibility results in this section are for forwarding patterns that do not match on the source.

\begin{lemma}\label{nok5}
The complete graph with five nodes does not allow for perfect resilience, i.e.,  $A_p(K_5,t)=\emptyset$.
\end{lemma}

\begin{proof}
	Let $V(K_5)= \{v_1,v_2,v_3,v_4,v_5\}$, where we assume w.l.o.g.\ $v_1$ to be the source $s$ and $v_5$ to be the target $t$.
	To prove the lemma, we construct sets of link failures in which we leave target $v_5$ connected to only one of the
  non-target nodes, and "fine-tune" the set of link failures so that $v_1$'s  packet will only visit 3 out of the 4 neighbors of $t$
	(say $(v_1,v_2,v_3)$).
  By contradiction, let $A\in A_p(K_5,t)$.
	%
	%
	Let $\pi^t_{v_1}(\cdot,F)$ be the port mapping produced by $A$ at node ${v_1}$ given $F=\{({v_1},t={v_5})\}$ and let  $\pi^t_{v_1}(\bot,F)=v_2$ w.l.o.g.
	Since $A$ is perfectly resilient, we know by Lemma~\ref{lemma:orbit} that $\pi_{v_1}^t$ is a permutation over its relevant neighbors (all neighbors are relevant and must be in the same orbit).
	Necessarily, ${v_2}$ has a predecessor\footnote{As all neighbors are relevant, the perfect resilient routing forms a cyclic permutation and we can directly identify the predecessor.} in its orbit in $\pi^t_{v_1}(\cdot,F)$.
	W.l.o.g.\ assume it is ${v_3}$: $\pi^t_{v_1}({v_3},F)={v_2}$.
	Construct further sets of link failures as follows:
  \begin{itemize}\setlength\itemsep{0em}
  \item $F_t=\{(t,{v_1}),(t,{v_2}),(t,{v_3})\}$: leave only ${v_4}$ connected to the
    target.
  \item $F_{v_2}=\{({v_2},t),({v_2},{v_4})\}$: make sure ${v_2}$ can only pass ${v_1}$'s
    packets to ${v_3}$.
    \item $F_{v_3}=\{({v_3},{v_4}),({v_3},t)\}$: make sure ${v_3}$ can only pass ${v_2}$'s
      packets to ${v_1}$.
    \end{itemize}

  Let $F_\emptyset= F\cup F_t\cup F_{v_2}\cup F_{v_3}$.
	Note that from the perspective of ${v_1}$, only the failure of $({v_1},t)$ is visible, and hence $F$ and $F_\emptyset$ are locally indistinguishable.
	Let us now construct the sequence of links traversed in $A$: at ${v_2}$, due to Lemma~\ref{lemma:orbit}, the packet is forwarded to ${v_3}$ (both $v_1,v_3$ are relevant for $v_2$).
	At ${v_3}$ it is necessarily forwarded to $v_2$ for the same reason.
	As $F_\emptyset$ and $F$ have the same impact on ${v_1}$'s ports,  we deduce $\pi^t_{v_1}({v_3},F)=\pi^t_{v_1}({v_3},F_\emptyset)={v_2}$.
	The link $({v_1},{v_2})$ is thus used repeatedly, i.e.,  $A$ causes a permanent loop.
  As $({v_1},{v_4})$ and $({v_4},t) \not \in F_\emptyset$, ${v_1}$ and $t$ are connected in $G\setminus F_\emptyset$, and yet $A$ loops, leading to the desired contradiction.
\end{proof}

\begin{lemma}\label{nok3}
The complete bipartite graph with six nodes, three in each partition, does not allow for perfect resiliency, i.e., it holds that  $A_p(K_{3,3},t)=\emptyset$.
\end{lemma}

\begin{proof}
  We proceed similarly to the $K_5$ case. Let
  $V_1=\{a,b,c=t\},V_2=\{v_1,v_2,v_3\}$ and $E=V_1\times V_2$, where we assume w.l.o.g.\ that we start on $a$.
  By contradiction, let $A\in A_p (K_{3,3})$.
	Let $v_1$ be the first target chosen by $a$.
	Let $\pi^t_a(\cdot,\emptyset)$ be the port mapping produced by $A$ at node $a$ given $F=\emptyset$.
	Since $A$ is perfectly resilient, we know by Lemma~\ref{lemma:orbit} that $\pi^t_a(\cdot,\emptyset)$ is a
  cyclic permutation over its relevant neighbors (all neighbors).
	Necessarily, $v_1$ has a predecessor in $\pi_a$.
	W.l.o.g.\ assume it is $v_2: \pi^t_a(v_2,\emptyset)=v_1$.

	We now construct a further set of link failures as follows, where we do not touch links incident to $a$:
  \begin{itemize}\setlength\itemsep{0em}
  \item $F_t=\{(t,v_1),(t,v_2)\}$: leave only $v_3$ connected to the
    target.
  \item $F_b=\{(b,v_3)\}$: make sure $b$ can only pass $v_1$'s packets to $v_2$ ($v_1,v_2$ are relevant).
  \end{itemize}
  Set $F'=  F_t\cup F_b$.
	Let us construct the sequence of links traversed in $A$:
	at $v_1$, due to Lemma~\ref{lemma:orbit}, the packet is necessarily forwarded to $b$.
	At $b$ it is necessarily forwarded to $v_2$ for the same reason.
	Since $F=\emptyset$ and $F'$ have the same impact on $a$'s forwarding function, we deduce
	$\pi^t_a(v_2,\emptyset) = \pi^t_a(v_2,F')=v_1$
	As link $(a,v_1)$ is used repeatedly, $A$ thus causes a permanent loop.
  Since $(a,v_3)$ and $(v_3,t) \not \in F'$ , $a$ and $t$ are connected in $G\setminus F'$, and yet $A$ fails, leading to the desired contradiction.
\end{proof}

We can now show that only planar graphs can permit perfect resilience.
In other words, if a graph is not planar, then it does not permit perfect resilience.

\begin{theorem}
\label{thm:planar-wagner}
  If $G$ is not planar, then it does not support a perfectly resilient forwarding pattern $\pi^t$, i.e., $A_p(G,t) = \emptyset$.
\end{theorem}
\begin{proof}
  First, observe that both $K_{5}$ and $K_{3,3}$ do not support perfectly resilient schemes due to Lemmas~\ref{nok5} and \ref{nok3}.
	%
	Next, Wagner's theorem~\cite{wagner} states that $G$ planar
  $\Leftrightarrow (K_5 \not\in min(G) \wedge K_{3,3} \not \in
  min(G))$.
	The contrapositive form of Corollary~\ref{minorstab} is $A(G',t)= \emptyset \Rightarrow A(G,t)=
  \emptyset$.
	As perfect resilience is impossible on both $K_5$ and $K_{3,3}$, we deduce that no graph with $K_5$ or $K_{3,3}$ as a minor permits a perfectly resilient scheme.
\end{proof}

\section{Negative Results}
\label{sec:feigenLim}

Our observations above also allow us to derive a number
of additional impossibility results.


\subsection{Boosting Feigenbaum et al.'s\ Impossibility} \label{ssec:replicate-feigenbaum}

Feigenbaum et al.\ \cite{podc-ba} gave a construction for which there is no perfect resilience forwarding pattern. We strengthen their result slightly, showing that no perfectly resilient forwarding pattern exists, even when the source is known. This allows us to boost their result to almost any resilience.
We defer the proof details to Appendix~\ref{app:fb}.

\vspace{1mm}

By padding and replicating Feigenbaum's construction we gain different parametrizations of the impossibility result. In particular, we observe that no $\omega(1)$-resilient forwarding pattern exists (as a function of $m$), and show that no $\Theta(f(n))$-resilient forwarding pattern exists even when the source and the target are $\Theta(f(n))$-connected. Theorems~\ref{thm:feigenbaum-padded} and \ref{thm:feigenbaum-padded-replicated} are asymptotic in $m$, the number of links in the input graph before failures.

\begin{theorem} \label{thm:feigenbaum-padded}
  There is no $\omega_m(1)$-resilient forwarding pattern with source, target, and in-port matching.
\end{theorem}

By a combination of padding and replicating the construction of Feigenbaum et al.\ it is possible to create a construction where the source and the target are connected by many link-disjoint paths after the failure, yet any forwarding pattern will fail.

\begin{theorem} \label{thm:feigenbaum-padded-replicated}
  There is no $\Theta(f(m))$-resilient forwarding pattern with source, target, and in-port matching, for any $f(m)$ between $\Omega(1)$ and $O(m)$, even when there is a promise that there are $\Theta(f(m))$ link-disjoint paths between the source and the target after the failures.
\end{theorem}

\subsection{Impossibility on Planar Graphs} \label{ssec:planar-negative}

In the previous section, we showed that 
non-planar graphs do not have resilient forwarding patterns.
We now show that there are also relatively small planar graphs that do not permit perfect resiliency.
Chiesa et al.~\cite{DBLP:journals/ton/ChiesaNMGMSS17} already showed the impossibility of perfect resilience on a planar graph with over 30 nodes with just two failures, but it is not clear how to extend their example to also account for the packet source.
However, as we will see later, for every outerplanar graph there is a perfectly resilient forwarding pattern, and there are also non-outerplanar planar graphs that allow for perfect resilience.

\begin{theorem} \label{thm:planar}
  There exists a planar graph $G$ on 7 nodes such that no forwarding pattern $\pi^t$ succeeds on $G$.
\end{theorem}

\begin{figure*}[t]
  \centering
  \includegraphics[width=0.55\columnwidth]{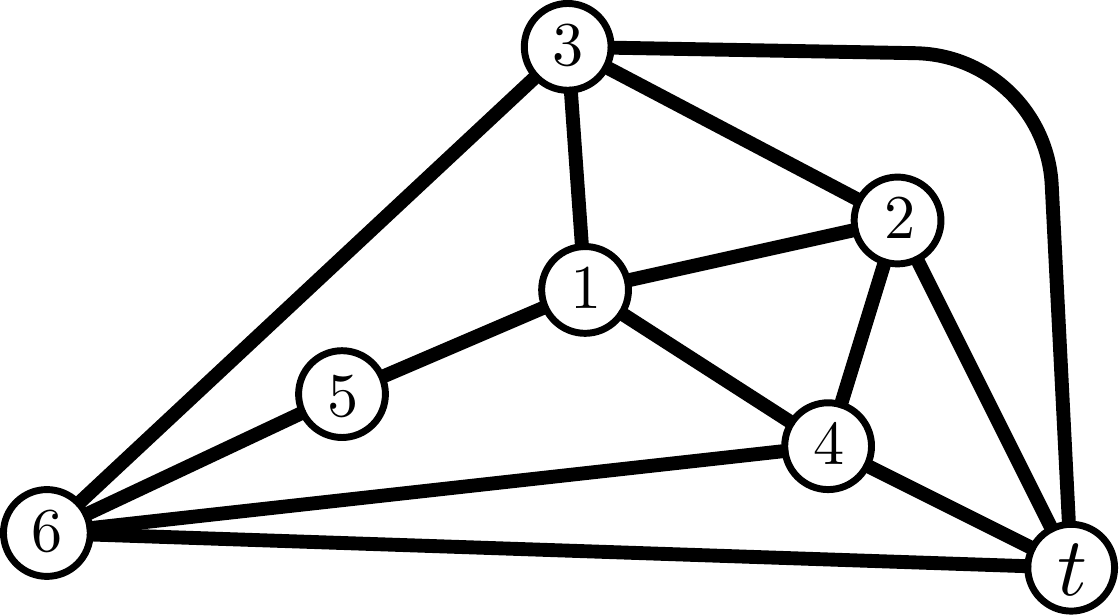}\hfill
  \includegraphics[width=1.4\columnwidth]{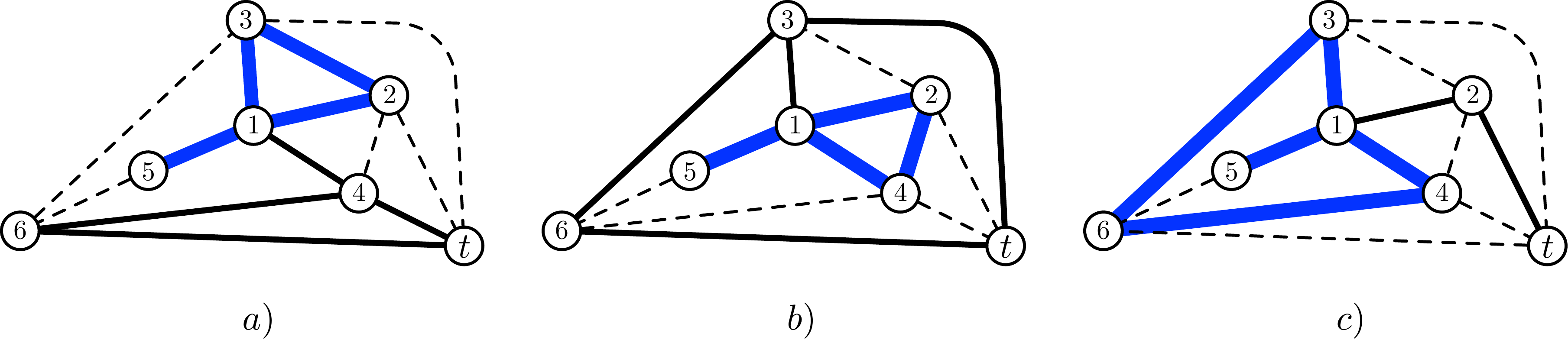}
	\vspace{-2mm}
  \caption{Planar graph without perfect resiliency.
	If the dashed links fail, in any forwarding pattern, packets will be stuck in one of the blue bold loops, even though there is at least one remaining path to the target.}
\label{fig:planar}
\vspace{-3mm}
\end{figure*}
\begin{proof}
Consider a packet emitted by node 5 for target $t$ in the graph $G$ depicted in Figure~\ref{fig:planar}.
In $G$ all neighbors of all nodes are relevant if no failures occur.
Thus the forwarding function at any node $v$ $\pi_v^t(\cdot, \emptyset)$ must be a cyclic permutation
of its neighbors due to Lemma~\ref{lemma:orbit}.
Since failures can only add relevant neighbors to nodes (Observation~\ref{obs:relevance}), any node of degree 2 of $G$ after failures must
forward packets with incoming port $p$ to port $p' \neq p$ due to the same lemma.
With this, we can show that there exists a failure set $F$ that leads to a loop for all possible forwarding
functions at node 1.
 We write  $\pi_1^t(\cdot, \emptyset)=(v_0, v_1, v_2, v_3)$ to denote the permutation that  assigns $v_{i+1 \mod 4}$ to $v_i$ for $v_i\in \{2,3,4,5\}$.
 We name the ports of node 1 by the identifiers of the neighbors they are connected to and we analyse
the different forwarding permutations at node 1, based on the out-port they assign to packets arriving on
in-port 5

Case (i)  $\pi_1^t(5, \emptyset)=2$. To ensure $\pi_1^t(\cdot, \emptyset)$ is a cyclic permutation under
this constraint, a packet arriving on port 2 can only be forwarded to either 3 or 4. Thus the possible cyclic
permutations are $(5,2,3,4)$  and $(5,2,4,3)$, which lead to loops under the failure sets illustrated in
Figure~\ref{fig:planar}.b  and a.

Case (ii)  $\pi_1^t(5, \emptyset)=3$. To ensure $\pi_1^t(\cdot, \emptyset)$ is a cyclic permutation under
this constraint, a packet arriving on port 4 can only be forwarded to either 2 or 5. Thus the possible cyclic
permutations are $(5,3,4,2)$  and $(5,3,2,4)$ which lead to loops under the failure sets illustrated in
Figure~\ref{fig:planar}.a  and c.

Case (iii)  $\pi_1^t(5, \emptyset)=4$. To ensure $\pi_1^t(\cdot, \emptyset)$ is a cyclic permutation under
this constraint, a packet arriving on port 4 can only be forwarded to either 2 or 5. Thus the possible
cyclic permutations are $(5,4,3,2)$  and $(5,4,3,2)$ which lead to loops under the failure sets illustrated in
Figure~\ref{fig:planar}.b  and c.

Case (iv)  $\pi_1^t(5, \emptyset)=5$ There is no cyclic permutation with this assignment, thus there would be a loop according to  Lemma~\ref{lemma:orbit}.

Hence there is exists no forwarding pattern without a failure set that causes a loop.
\end{proof}

We can extend the proof with Lemma~\ref{lem:s-orbit} and Corollary ~\ref{corr:s-paths-deg} to also include the source:

\begin{theorem} \label{thm:planar-s}
  There exists a planar graph $G$ on 8 nodes s.t.\ no forwarding pattern~$\pi^{s,t}$ succeeds on $G$.
\end{theorem}

\ifExtended
The full proof is deferred to \cref{sec:planar-s}.
\fi
\ifFull
\begin{proof}
We note that by choosing the node 5 in Figure~\ref{fig:planar} as the source, the proof for Theorem~\ref{thm:planar} would directly carry over, if we could apply Lemma~\ref{lemma:orbit} to all nodes in Figure~\ref{fig:planar}.
However, Lemma~\ref{fig:planar} only considers target-based routing that does not consider the source, and we hence need to utilize \cref{lem:s-orbit} for node 1 (all relevant neighbors are in the same orbit) and \cref{corr:s-paths-deg} for the other nodes, where the degree two case suffices.
In order to apply Lemma~\ref{lem:s-orbit} to node 1, we need to connect the new source $s$ to two relevant neighbors of $s$, where we pick nodes 3 and 5, which keeps the 7+1 node construction planar.
Note that for the actual packet path in Figure~\ref{fig:planar}(a,b,c), we can fail the link $(s,3)$, keeping $(s,5)$ alive, and if $5$ were to route back to $s$, $s$ can only bounce the packet back.
It remains to force all degree 2 nodes in Figure~\ref{fig:planar}(a,b,c) to relay incoming packets through the other port, via Corollary~\ref{corr:s-paths-deg}, which we prove by case distinction.
To this end, for each of the seven cases (nodes 2,3 for a, 2,4 for b, and 3,4,6 for c), we need to show that there are pairs of node-disjoint paths for each such node $i$, with two different surviving links $(v_1,i),(v_2,i)$, that are in $E(G)$ minus all links incident to $i$, except $(v_1,i),(v_2,i)$:\footnote{Recall that $i$ is only aware of its incident link failures.}
\begin{itemize}\setlength\itemsep{0em}
	\item start from $s$ and end with $(v_1,i)$, and start with $(v_2,i)$ and end at $t$,
	\item start from $s$ and end with $(v_2,i)$, and start with $(v_1,i)$ and end at $t$.
\end{itemize}

\vspace{-3mm}

\begin{itemize}\setlength\itemsep{0em}
	\item Figure~\ref{fig:planar}(a): Node 2 $\left[(3,2),(2,1) \right]$:  $s-5-6-3-2$ \textbf{:} $2-4-t$ and $s-5-1-2$ \textbf{:} $2-3-t$
	\item Figure~\ref{fig:planar}(a): Node 3 $\left[(1,3),(3,2) \right]$: $s-5-1-3$ \textbf{:} $3-2-t$ and $s-5-1-2-3$ \textbf{:} $3-6-t$
	\item Figure~\ref{fig:planar}(b): Node 2 $\left[(1,2),(2,4) \right]$:
		 $s-5-1-2$ \textbf{:} $2-4-t$ and $s-5-1-4-2$ \textbf{:} $2-1-t$
	\item Figure~\ref{fig:planar}(b): Node 4 $\left[(2,4),(4,1) \right]$: \item $s-5-1-4$ \textbf{:} $4-2-3-t$ and $s-3-2-4$\footnote{Note that the packet is not actually routed via $3$ in Figure~\ref{fig:planar}(b), but the node 4 must also provision for this case in order to guarantee perfect resilience.} \textbf{:} $4-1-5-6-t$
	\item Figure~\ref{fig:planar}(c): Node 3 $\left[(1,3),(3,6) \right]$: $s-5-1-3$ \textbf{:} $3-6-t$ and $s-5-6-3$ \textbf{:} $3-1-4-t$
	\item Figure~\ref{fig:planar}(c): Node 4 $\left[(1,4),(4,6) \right]$: $s-5-1-4$ \textbf{:} $4-6-t$ and $s-5-6-4$ \textbf{:} $4-1-3-t$
	\item Figure~\ref{fig:planar}(c): Node 6 $\left[(3,6),(6,4) \right]$: $s-5-1-3-6$ \textbf{:} $6-4-t$ and $s-5-1-4-6$ \textbf{:} $6-3-t$
\end{itemize}
\vspace{-3mm}
\end{proof}
\fi

\vspace{-1mm}

\section{Positive Results}
\label{sec:positive}

Despite the numerous networks in which perfect resilience
cannot be achieved, there are several interesting scenarios for which
perfectly resilient algorithms exist.
In this section, we present a particularly simple
algorithmic technique: the algorithm orders its neighbors arbitrarily, 
and reacts to link failures by simply ``skipping over'' the failed links according to this order. 
The exact forwarding function can hence be computed locally based on $F$, and the forwarding rules can be stored in linear space: a significant advantage in practice. 
This motivates us to introduce the following definition:

\begin{Definition}[Skipping Forwarding Functions]
Given a set $S$ and a function $\pi \colon S \mapsto S'$ where $S' \subset S$, define the \emph{tail} of $s$ to be the sequence $(\pi(s), \pi(\pi(s)), \pi(\pi(\pi(s))), \dots )$ for each $s \in S$. We say that a forwarding function $\pi_v^{s,t}$ is \emph{skipping} if there exists a bijection $f \colon E(v) \cup \{ \bot \} \mapsto E(v)$ such that for each failure set $F$ and each $e \in \{ \bot \}\cup E(v) \setminus F$ we have that $\pi_v
^t(e,F)$ equals the first element in the tail of $e$ (with respect to $f$) that is not in $F$. A forwarding pattern is skipping if each of its forwarding functions is skipping.
\end{Definition}

The positive results in the following subsections will all rely on such skipping.

\subsection{The Target is Close}

Assume that a source $s$ is close to the target $t$, even after the failures.
In the following, we show that it is always possible to predefine conditional forwarding rules, with skipping, which ensure a route from $s$ to $t$ if their distance is at most two hops.

\begin{theorem}\label{thm:2hops}
For all graphs $G$ there is a forwarding pattern, matching on the source, that succeeds if source $s$ and target $t$ are at distance at most 2 in $G \setminus F$.
\end{theorem}

\ifExtended
\fi
\ifFull
\begin{proof}
Resilience can be ensured with the following forwarding pattern.
Let $V(s)$ be the set of neighbors of the source node $s$
which are on a path of length two to the target $t$ before
the failures.
We define the forwarding function $\pi^{s,t}_s$ of the source $s$ as follows:

\begin{itemize}\setlength\itemsep{0em}
\item We order the neighbors $V(s)$ of $s$ arbitrarily, i.e.,
$V(s)=(v_1,v_2,\ldots,v_k)$, and source $s$ first tries to forward
to the first neighbor $v_j$ which is connected to it after the failures ($\pi_s^{s,t}(\bot) = v_j$), i.e., in a skipping fashion.

\item When a packet arrives on the in-port from $v_i$, for any $i$,
$s$ forwards the packet to $v_{j}$, the next neighbor in $V(s)$ such that $\{s, v_j\} \notin F$, i.e., skipping as well.
\end{itemize}

For all other nodes $v$, the forwarding function is defined as follows.
When a packet arrives from the source $s$, send it to the target if the link has not failed. Otherwise return the packet to the source.
Note that this is a skipping pattern, as we can predefine a skipping of $(v,t)$ to route to $s$.
This ensures that the source $s$ tries all of its neighbors which could be
on a path of length two to the target. Since we assumed such a path existed, the forwarding pattern must succeed in routing the packet to the target.
\end{proof}
\fi

The algorithm relies on the source probing all neighbors if they are connected to the destination, which then succeed or return the packet. 
We can extend this idea to not require matching on the source, if the target is at most 2 hops away after failures. 
The idea is to assign unique identifiers to the nodes, converge to a 2-hop minimum identifier, and then forward similarly to the previous proof.

\begin{theorem}\label{thm:2hops-no-s}
There is a forwarding pattern, not matching on the source, that ensures a packet reaches its target $t$ if all nodes are at most 2 hops apart from $t$ after failures.
\end{theorem}

\ifExtended
The full proofs are deferred to \cref{sec:2hops-no-s,sec:2hops}.
%
\fi
\ifFull
\begin{proof}
Let us assign a unique ID to every node and construct the forwarding pattern based on these identifiers.
The forwarding function of each node $v \in V \setminus \{s,t\}$ is defined as follows:
\begin{itemize}\setlength\itemsep{0em}
	\item If $t$ is a neighbor of $v$, forward to $t$.
	\item If $t$ is not a neighbor, then
	\begin{itemize}\setlength\itemsep{0em}
		\item if $v$ does not have the lowest ID in its neighborhood, forward to the neighbor with lowest ID, i.e., skipping potential lost neighbors with lower ID,
		\item if $v$ has the lowest identifier in its neighborhood, then route according to a skipping pre-defined cyclic permutation of all neighbors.
	\end{itemize}
\end{itemize}
Observe that, unless the target is reached on the way, the packet will reach a node $w$ that has the lowest identifier in its two-hop neighborhood.
From there on, the argument is analogous to the proof of Theorem~\ref{thm:2hops}:
The node $w$ will forward the packet to all its one-hop neighbors, which in turn bounce it back to $w$, unless connected to $t$. Since, by assumption, one of them is connected to $t$, the forwarding succeeds.
\end{proof}
\fi

\subsection{Planar and Outerplanar Graphs}\label{subsec:outerplanar}

A graph is outerplanar if there exists a planar embedding such that all nodes are part of the outer face.
In other words, there is a walk along the links of the outer face that visits all nodes for connected graphs.
This property holds also after arbitrary failures as long as the graph remains connected.
Thus we can route along the links of the outer face of a planar graph using the
well known right-hand rule~\cite{reason:AbeDiS81a,Bondy1976} despite failures.\footnote{Such face routing was first considered over 20 years ago for ad-hoc networks by Kranakis et al.~\cite{DBLP:conf/cccg/KranakisSU99} and Bose et al.~\cite{DBLP:conf/dialm/BoseMSU99}. More involved face-routing algorithms have been devised and analyzed e.g.~\cite{DBLP:conf/sirocco/WattenhoferWW05,DBLP:journals/tc/FreyS10} and go beyond the simple right-hand rule algorithm used for Theorem~\ref{thm:outerplanar}
and Corollary~\ref{cor:sameface}.
They have mostly been studied without considering failures and use either state in packet or at nodes
if source and target are not on the same face. We refer to some reference articles~\cite{DBLP:reference/algo/Zollinger16,DBLP:journals/ton/KuhnWZ08} and book chapters~\cite{doi:10.1002/047174414X.ch12,doi:10.1002/0470095121.ch11} for in-depth discussions and to the article by Behrend~\cite{labyrinth} for a historical overview.}

\begin{theorem}\label{thm:outerplanar}
Let $G=(V,E)$ be an outerplanar graph. 
Then, there is a perfectly resilient skipping forwarding pattern $\pi^t$ which does not require source matching.
\end{theorem}

\ifExtended
\fi
\ifFull
\begin{proof}
Fix an arbitrary outerplanar embedding of the graph without failures (a consistent 
embedding must be used on all nodes when constructing the forwarding pattern). 
Without loss of  generality we define the canonical direction to be clockwise and we 
number the ports of the nodes accordingly starting from an arbitrary~port. 

If a link $l$ belongs to the outer face of a planar graph $G$, it also belongs
to the outer face for all subgraphs of $G$, in particular also in $G\setminus F$ for $l \notin F$. Since 
all nodes belong to the outer face for an outerplanar graph, it is hence enough to 
demonstrate that the forwarding pattern ensures packets use only the links of the 
outer face and do not change the direction despite failures.

Given an in-port belonging to a link on the outer face of a graph in clockwise 
direction we can decide locally which out-port belongs to the outer face to continue 
the walk using the right-hand rule encoded in the the forwarding function with $F$.

The routing algorithm assigns the following forwarding functions to node $v$ with 
degree $d$: $\pi_v^t(x) = x +1 \bmod d$. Hence the forwarding
function assigns out-port~$x +1 \bmod d$  to in-port $x$.
When node $s$ sends a packet to $t$ it picks an arbitrary out-port
that belongs to the outer face in the clockwise direction.
When failures occur, skipping is used to derive the next usable
output. 

As we will see next, this is enough to encode the information of which links belong 
to the outer face of $G\setminus F$. To this end, we show that this forwarding pattern 
routes a packet along the links of the outer face for an arbitrary link failures $F$ in 
clockwise direction.

We show this claim for the source and then for all other nodes on the
path to the target by~induction.

Each node has at least one pair of an incoming port and an outgoing
port on the outerplanar face of $G$ (for a node with degree one these two
ports belong to the same link).  Let $(i_j, o_j)$ denote the $j^{th}$
such pair out of $k$ for node $v$ in $G$, starting to number them from
port 0 at node $v$ in clockwise order.
Note that it holds for all these $k$ pairs that $o_j = i_j + 1 \bmod d$.

Without failures, the source selects an outgoing port on the outer face in
the correct direction by definition. Without loss of generality, let this out-port be
$o_0$ from the pair $(i_0, o_0)$ defined above. If the corresponding link
has failed, the routing scheme will try port $o_0 + j \bmod d$ for
$j = 1, \dots, d-1$ until successful according to the skipping forwarding pattern.
Due to the failures, the corresponding link is now on the outer face and the
statement thus holds for the base case.

For a subsequent node $v$ on the walk to the target the packet will enter the
node on  port $i_j$  for some $j$ if no failures affect this node. In this case,
the routing scheme selects $i_j + 1 \bmod d$ as the out-port, which
happens to be $o_j$ as discussed above. Thus the packet remains using only links
of the outer face travelling in the right direction on $G\setminus F$. If there has been a failure affecting
node $v$, then the outgoing port might no longer be available. As shown for the source,
iterating over the next ports clockwise will guarantee that the packet is still using a link
of the outer face of $G\setminus F$. If the link incident to $i_j$ is down, then $v$'s neighbors
will select another link and the packet might enter the node on a port that
is not among the pairs for the graph without failures. Nevertheless, due to
the arguments above this implies that the incoming port now is part of
the walk along the links of the outer face in clockwise direction
 and also the outgoing port
chosen subsequently is part of the links on the outer face of $G\setminus F$.
Hence the routing scheme succeeds to lead packets to the target successfully
 without maintaining state in packets or at nodes.
\end{proof}
\fi

The face-routing pattern is even target-oblivious: starting on any node, it will visit every node.
Moreover, while we have seen earlier that planar graphs do not offer perfect resilience, the approach guaranteeing perfect resilience for outerplanar graphs can also be applied for planar graphs if source and  target are on the same face.

\begin{corollary}\label{cor:sameface}
%
Let $G=(V,E)$ be a planar graph where packets start on the same face as $t$.
Then, there is a perfectly resilient skipping forwarding pattern $\pi^{t}$.
\end{corollary}

\ifExtended
The full proofs are deferred to \cref{sec:outerplanar,sec:sameface}. 
\fi
\ifFull
\begin{proof}
The algorithm from the proof for Theorem~\ref{thm:outerplanar} only needs to be
adapted at the start. Since the source and the target might not be
part of the outer face in $G$, the source chooses an arbitrary (clockwise direction) out-port that belongs to the face it shares with the target (some nodes may share several faces with the target, any of them can be chosen in this case).
Since failures only reduce the number of faces, the shared face may contain more
links after failures but not fewer and both source and target will always
belong to the same face on $G\setminus F$, as long as the graph is still connected.
Thus following the links of the face in clockwise direction will lead to the target
for arbitrary non-disconnecting failures.
\end{proof}
Another way of thinking about the proof is that we can also choose the embedding of the planar graph s.t.\ any desired face is the outer face, see e.g.\ Schnyder~\cite{DBLP:conf/soda/Schnyder90}.
As before, note that our forwarding pattern can visit every node on the face despite failures. 
\fi

\subsection{A Non-Outerplanar Planar Graph with Perfect Resilience: $K_4$}

So far, we established that perfect resiliency is possible on outerplanar graphs as well as on the same face of planar graphs, and that it is impossible on some planar graphs and on all non-planar graphs.
This raises the question if perfect resilience is possible on some non-outerplanar planar graphs, which we answer in the affirmative for $K_4$, the complete graph with four nodes: we employ forwarding along a cyclic permutation, unless the target is a neighbor.

\begin{theorem}\label{thm:k4}
$K_4$ allows for perfectly resilient forwarding patterns $\pi^t$ with skipping, i.e., without the source. 
\end{theorem}

\ifExtended
The full proof is deferred to \cref{sec:thm:k4}.
\fi
\ifFull
%
%
\fi

%
\begin{proof}
From \S\ref{subsec:outerplanar}, we know that we can visit every node on a connected outerplanar graph using the skipping right-hand rule.
Hence, if a graph becomes outerplanar after removal of the destination, as $K_4$ does, it allows for perfect resilience, by traversing all nodes, each time checking if they neighbor the destination.
\end{proof}

\begin{corollary}
Let $G'=(V\setminus\{t\},E)$ be outerplanar. 
Then $G=(V,E)$ allows for perfectly resilient forwarding patterns $\pi^t$ with skipping, i.e., without the source.
\end{corollary}


\subsection{Graph-Families Closed under Subdivision} \label{sec:subdivision}

In this section we prove that under certain conditions, skipping forwarding functions retain any resilience guarantees.
Our result applies to any graph family~$\G$ that is closed under link subdivision. A family~$\G$ is closed under link subdivision if for all $G \in \G$ the graph $G'$ constructed by replacing a link $(u,v)$ of $G$ by a new node $w$, and two new links $(u,w)$ and $(w,v)$ is also in $\G$.%
\footnote{Chiesa et al.~\cite[Fig.~9(b)]{DBLP:journals/ton/ChiesaNMGMSS17} used a similar subdivision idea, where each link was replaced with three links and two new nodes, where we use two links and one new node. However, Chiesa et al.\ used this idea in a different context, namely to show the impossibility of perfect resilience on a planar graph example.}
For example planar and bounded genus graphs are closed under subdivision, but graphs of bounded diameter are~not.

\begin{theorem}\label{thm:sim-argument}
 Let $\G$ be a family of graphs that is closed under the subdivision of links. Assume that for each $G \in \G$ there exists a $f(m)$-resilient forwarding pattern (with or without source matching). Then there exists a $f(m)$-resilient skipping forwarding pattern with the same matching for each $G \in \G$, where $f$ is any monotone increasing function.
\end{theorem}

\begin{proof}[Proof of Theorem~\ref{thm:sim-argument}]
  Consider any $G \in \G$ and let $|E(G)| = m$. Consider a graph $H$ that is constructed by subdividing each link $e = (u,v) \in E(G)$ into three links $e_u = ( u, uv ), e_v = ( v, vu ),$ and $e_{uv} = ( vu, uv )$, where $uv$ and $vu$ are the two new nodes introduced in the subdivision. We call these nodes the \emph{new nodes} and the other nodes \emph{old nodes} of $H$. Let~$M = |E(H)|$. We prove the theorem for forwarding patterns with source matching -- the case with source matching is similar.

  Now given a set of failed links $F \subseteq E$ on $G$, we can consider the corresponding failure set $F' = \{ (uv, vu ) \in E(H) : (u,v) \in F \}$ in $H$. Since $m \leq M$ and therefore $f(m) \leq f(M)$, there must exist a $f(M)$-resilient forwarding pattern $\phi^t$ for $H$ and any $t \in V(H)$. Since only the middle links of the subdivided links fail, the forwarding functions of the old nodes must remain constant over all $F'$ constructed from the failure sets $F$. 

  Before simulating $\phi^t$ on $G$, we take care of a technicality. We say that a link $(u,uv) \in E(H)$ is \emph{cut} by $\phi^t$ if, for all $F$ not containing $(u,uv)$, $(uv,vu)$, or $(vu,v)$, either $\phi^t_{uv}(u,F) = u$ or $\phi^t_{vu}(uv,F) = uv$. That is, $uv$ or $vu$ sends the packet from the direction of $u$ back. In this case the packet cannot pass from $u$ to $v$. We modify $\phi_u^t$ to ignore cut links: let $S$ denote the set of cut links in $E(u)$. For each $e \in E(u)$ such that $\phi_u^t(e) \in S$ we set $\phi_u^t(e)$ to be the first element in the tail of $e$ not in $S$.

  Now let $\phi^t$ be such a modified forwarding pattern for any old node $t \in V(H)$. We construct a forwarding pattern $\pi^t$ for $G$ as follows. For each $F$, each $v \in V(G)$, and each $e$ incident to $v$, set $\pi^t_v(e,F)$ to be the first element in the tail of $e$ that is not in $F$. In the degenerate case where $e = E(v) \setminus F$ assign $\pi_v^t(e,F) = e$. Finally, if all links corresponding to non-failed links are cut, we set $\pi_v^t(e,F) = e$ for each non-failed cut link $e$. In this case the forwarding pattern will never reach $v$ in $H$.

  By construction the forwarding pattern $\pi^t$ is skipping. Links corresponding to the cut links of $H$ are never forwarded to and failed links are skipped. It remains to show that it is correct. Let $F \subsetneq E(G)$ be a failure set such that $|F| \leq f(n)$. At each node $v \in V(H)$ with in-port $e$, the packet will travel the tail of $e$ until it finds an out-port that has not failed (i.e.\ the middle link of the subdivided link does not correspond to a failed link of the original failure set $F$). The packet, due to the modification, will then travel over the subdivided links $(v,vu)$, $(vu, uv)$, and $(uv, u)$ to some old node $u$. In the original graph, due to the construction of $\pi^t$, the packet with in-port $e$ will directly be forwarded to $(v,u)$. By an inductive argument for each initial node and failure set $F$, the sequence of nodes followed by the packet in $G \setminus F$ equals the sequence of \emph{old nodes} followed by the packet in $H \setminus F'$.
\end{proof}

\noindent It follows that in many graph classes, it is sufficient to consider skipping forwarding patterns.

\begin{corollary} \label{cor:sim-argument}
  In the following graph classes, there is a resilience-optimal skipping forwarding pattern: planar graphs, graphs of bounded genus, cycles, graphs of bounded maximum degree, and of bounded arboricity.
\end{corollary}

\section{Conclusion}
\label{sec:conclusion}

We studied the fundamental question of when it is possible
to provide perfect resilience in networks based on local
decisions only. We provided both characteristics of
infeasible instances and algorithms for robust networks for perfect 
and parametrized resilience.

While our results cover a significant part of the problem space,
it remains to complete charting the landscape of the feasibility
of perfect and parametrized resilience, both in the model where the source can
and cannot be matched. Furthermore, we have so far focused
on resilience only, and it would be interesting to
account for additional metrics of the failover paths,
such as their length and congestion.

\vspace{-2mm}

\subsection*{Acknowledgements} 
We would like to thank Jukka Suomela for several fruitful discussions.
We would also like to thank the anonymous reviewers and Wenkai Dai for their helpful comments.

\ifExtended
\clearpage
\fi
\ifFull
\fi

{
\small
\setlength{\bibsep}{0pt plus 0.3ex}
\bibliographystyle{plainnat}
\bibliography{literature}
}

\appendix

\section{Proofs for Section~\ref{ssec:replicate-feigenbaum}}\label{app:fb}

\begin{figure*}[t]
  \centering
  \includegraphics[width=0.6\textwidth]{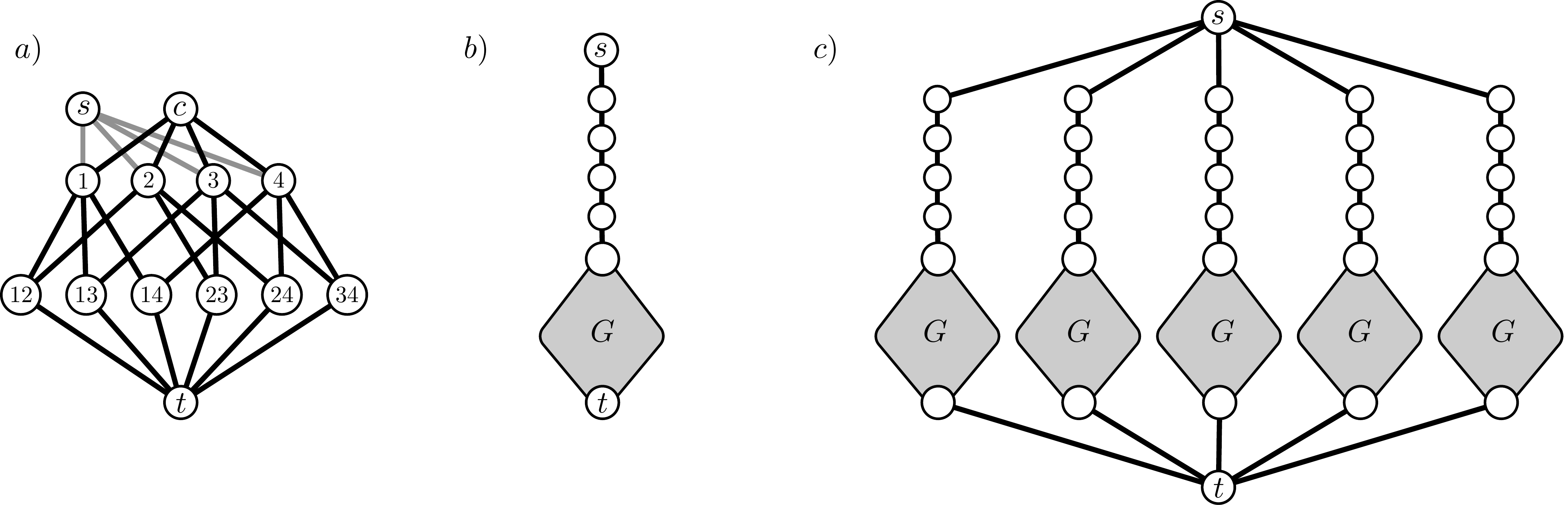}
	\vspace{-5mm}
  \caption{a) The original construction of Feigenbaum et al.\ \cite{podc-ba} with an extra source node $s$. b) The padded construction for Theorem~\ref{thm:feigenbaum-padded}. $G$ represents the gadget from Theorem~\ref{thm:feigenbaum-gadget}. c) Padded and replicated construction from Theorem~\ref{thm:feigenbaum-padded-replicated}. Each gadget $G$ is a copy of the gadget from Theorem~\ref{thm:feigenbaum-gadget}.} \label{fig:feigenbaum-gadget}
	\vspace{-5mm}
\end{figure*}

\begin{theorem} \label{thm:feigenbaum-gadget}
  There exists a graph $G$ on 13 nodes such that no forwarding pattern will succeed on $G$ when source, target, and in-port are known.
\end{theorem}
\vspace{-2mm}
\ifExtended
For the sake of completeness, we present a proof of Theorem~\ref{thm:feigenbaum-gadget} in \cref{sec:feigenbaum-proof}.
\fi
\ifFull
\begin{proof}[Proof of Theorem~\ref{thm:feigenbaum-gadget}]
  Construct the gadget $G$ as follows. Take four nodes 1, 2, 3, and 4 and connect them by a link to a center node $c$. Then, for each unordered pair $(i,j)$ of nodes from $\{1,2,3,4\}$ create a new node $ij$, and connect $ij$ to $i$ and $j$. Connect all $ij$ to a new target node $t$. Finally, create a source node $s$ and connect it to $\{1,2,3,4\}$. $G$ is illustrated in Figure~\ref{fig:feigenbaum-gadget}a.

  We will call nodes 1, 2, 3, and 4 \emph{level one} nodes, and nodes 12, 13, 14, 23, 24, and 34 \emph{level two} nodes.
  We first claim that if nodes of levels one and two have degree 2 after failures, then they must \emph{always} forward the packet coming from $s$, level one, or level two to the other port, with the exception of forwarding back towards $s$.
  \begin{itemize}\setlength\itemsep{0em}
    \item If node $vu$ on level two is connected to $t$, we can cut all other links to $t$ and force $vu$ to forward to $t$, as otherwise the forwarding fails.
    \item If $uv$ is not connected to $t$, we can again create a unique path to $t$ that goes through $uv$. Choose some $vw$ on level two and cut all links to $t$ except $\{vw,t\}$. In addition cut all links to $vw$ except $\{v,vw\}$, and cut all links to $v$ except $\{uv, v\}$. A packet coming from $\{u,uv\}$ must be forwarded to $\{uv,v\}$. The other direction is symmetric.
    \item If a level one node $v$ is connected to $s$, then by setting $F = \{ \{s,u\} : u \neq v \}$ any non-forwarding rule at $v$ fails.
    \item If a level one node $v$ is not connected to $s$, then it is connected to some level two nodes $uv$ and $wv$ or $c$ and some level two node $uv$. In the first case, by cutting all links from level two to $t$ except for $uv$ or $wv$, and, respectively, the link $\{u, uv\}$ or the link $\{w,wv\}$, we force the unique path from $s$ to $t$ to take the links $\{vu, v\}$ and $\{wv, v\}$, in both directions. In the second case, we can create the following unique path from $s$ to $t$ by failing all other links: $(s, u, uv, v, c, w, wu, t)$, where $w$ is some node not $v$ or $u$. This forces forwarding at $u$ from $uw$ to $c$. For the other direction fail everything except the path $(s, u, c, v, vw, t)$ for some $w$.
  \end{itemize}
  Any perfectly resilient forwarding pattern must therefore forward packets over degree-2 nodes.

  We assume that there are no failures around $c$ so the forwarding function $\pi_c^{s,t}$ remains fixed. In the following we identify the incident links of $c$ with the neighbors they connect to. If $\pi_c^{s,t}(i) = i$ for any $i \in [4]$ the forwarding fails when $s$ is connected only to $i$ and $i$ only to $c$. Assume without loss of generality that $\pi_c^{s,t}(1) = 2$. By having $\pi_c^{s,t}(2) = 1$ we have a loop when all other links incident to 2 fail and 1 has links $\{s,1\}$ and $\{1,c\}$.  Without loss of generality assume that $\pi_c^{s,t}(2) = 3$. By a similar argument we see that we must have $\pi_c^{s,t}(3) = 4$, as otherwise we would create a loop. Now consider different settings for $\pi_c^{s,t}(4)$.
  \begin{itemize}
    \item $\pi_c^{s,t}(4) = 1$: we fail all links incident to level one except $\{s,4\}, \{1,13\}, \{13,3\}$, all links incident to $c$ or 2. The forwarding pattern loops $(s, 4, c, 1, 13, 3, c, 4)$ and then either visits $s$ or goes back to $c$ immediately, finishing the loop.
    \item $\pi_c^{s,t}(4) = 2$ or $\pi_c^{s,t}(4) = 3$: These cases are similar, as we can force the packet to never visit 1, which could be the unique remaining path to $t$. Fail all links of $s$ except $\{s,4\}$, and all other links of 4 except $\{4,c\}$. In addition fail all links from 2 and 3 to level two. The packet will start a loop on $(s, 4, c, 2, c, 3, c, 4)$ or on $(s, 4, c, 3, c, 4)$, respectively. Then it will either visit $s$ or go directly back to $c$, finishing the loop.
  \end{itemize}
  As all choices for $\pi_c^{s,t}(4)$ create a forwarding pattern with a loop, one cannot obtain perfect resilience.
\end{proof}
\fi

\balance

By padding and replicating Feigenbaum's construction we gain different parametrizations of the impossibility result. In particular, we observe that no $\omega(1)$-resilient forwarding pattern exists (as a function of $m$), and show that no $\Theta(f(n))$-resilient forwarding pattern exists even when the source and the target are $\Theta(f(n))$-connected. Theorems~\ref{thm:feigenbaum-padded} and \ref{thm:feigenbaum-padded-replicated} are asymptotic in $m$, the number of links in the input graph before failures.

\begin{proof}[Theorem~\ref{thm:feigenbaum-padded}]
  Modify graph $G$ as follows: replace the source node $s$ with a path $P_k = (s_0, s_1, \dots, s_k)$ on $k$ nodes, and connect node $s_k$ to the nodes 1, 2, 3, and 4. The node $s_0$ is the new designated source node. Denote this construction by $G_k$. See Figure~\ref{fig:feigenbaum-gadget}b for an illustration.

  Any forwarding pattern for $G$ could simulate the existence of the path $P_k$. Since there is no forwarding pattern for $G$, there is no forwarding pattern for $G_k$. 
%
By setting $k$ large enough we can make $|E(G)| = f(k)$ for $f(k)>18$, for any function $f(m) = \omega_m(1)$.
\end{proof}
\begin{proof}[Theorem~\ref{thm:feigenbaum-padded-replicated}]
  First observe that the gadget $G$ contains $m_0 = 22$ links. The gadget $G_k$ from Theorem~\ref{thm:feigenbaum-padded} contains a total of $m_k = k+22$ links.

  The case of $f(m) = O(1)$ is covered by Theorem~\ref{thm:feigenbaum-gadget}. Therefore fix $f(r)$ to be a function between $\omega_r(1)$ and $r$. We construct an infinite family of graphs as follows. For a fixed $r$, take $f = f(r)$ copies $G^{(1)}_g$, $G^{(2)}_g, \dots, G_g^{(f)}$ of $G_g$ for $g = g(r) = \lceil r / f(r) \rceil$. Then identify the nodes $s_0^{(1)}, \dots, s_f^{(1)}$ as the same node, which is designated to be the source node. Finally, add a new designated target node $t$ and connect it to the old target nodes $t^{(1)}, \dots, t^{(f)}$ in each copy of the gadget. See Figure~\ref{fig:feigenbaum-gadget}.c) for an illustration.

  Since the links of the different gadgets can fail independently, and, by Theorem~\ref{thm:feigenbaum-gadget}, for each $G^{(i)}$ and each partial forwarding pattern of $G^{(i)}$ there exists a partial failure set $F$ that prevents the packet from being forwarded to the (old) target node $t^{(i)}$. By assigning the suitable failure sets to each gadget we can see that there exists a path between $s$ and $t$ through each gadget $G^{(i)}$ yet no packet will ever reach an old target node $t^{(i)}$ and therefore cannot reach the target node $t$.

  By construction there are $\Theta(r)$ links and $\Theta(f(r))$ paths between $s$ and $t$ after failures. In each gadget at most a constant number of links fail (inside each original gadget $G^{(i)}$), and therefore the total number of failed links is $\Theta(f(r))$, as required.
\end{proof}

\ifExtended

\appendix

\clearpage

\section{Proof of \cref{cont2}} \label{sec:cont1}
We will utilize the following lemma for the proof of \cref{cont2}:

\begin{lemma}[Contraction Stability: Algorithm Transfer]
\label{cont1}
Given $G$ and $(i,j)\in E(G)$ let $G'$ be the corresponding $(i,j)$-contraction.
	Let $R=\{(j,r),r\in  V(i)\cap V(j)\}$.
	Let $A:F\mapsto \{\pi^t_v(\cdot,F),v\in V\}$.
	Define $A': F\mapsto  \{\pi^t_v(\cdot,F),v\in V'\}$ to be the \emph{$(i,j)$-contracted algorithm} of $A$ as follows:
  \begin{itemize}
  \item Case \RomanNumeralCaps{1}: Identical behavior on unaffected nodes, $\forall v\in V', v\neq \{i\}$, $\pi'^t_v(\cdot,F)=\pi^t_v(\cdot,F\cup R)$.
  \item Case \RomanNumeralCaps{2}: Replace $i$'s permutation by the contracted algorithm. Let $\pi'^t_i(\cdot,F)=\pi^t_{\{i,j\}}(\cdot,F\cup R)$.
  \item Case \RomanNumeralCaps{3}: Replace $j$'s port by $i$'s port on $j$'s neighbors
    permutations: $\forall k\in V(j),\exists v \text{~s.t.~} \pi^t_k(v,F \cup R)=j
    \Rightarrow \pi'^t_k(v,F)=i$,  $\forall k\in V(j),\exists v \text{~s.t.~} \pi_k^t(j,F \cup R)=v
    \Rightarrow \pi'^t_k(j,F)=v$.
  \end{itemize}
Let $P$ (resp $P'$) be the sequence of links traversed using $A$ by a
packet from $s$ to $t$ under $F\cup R$ (resp traversed using $A$'
under $F$). 

Let $Q$ be a rewriting of $P$ in which we replace every occurence of $j$
by $i$. And let $Q'$ be the rewriting of $Q$ in which we remove every
occurrence of $(i,i)$. We have $Q=P'$.
\end{lemma}

\begin{proof}[Proof of Lemma~\ref{cont1}]
  We proceed by induction on the $k$ first hops  $Q[1..k]$ and
  $P'[1..k]$ of the sequences. When the context is clear, we write
  $\pi_v$ for the forwarding function of node $v$ using $A$ in context $F\cup R$,
  and $\pi'_v$ the forwarding function of $v$ using $A'$ in context $F$.

\emph{  Base case}: let $Q[1]=(s,q),P'[1]=(s,a)$. We need to prove
  $q=a$. $1)$: if $s\neq i$, $\pi'_s(\bot)=\pi_s(\bot)$ as defined in case
  \RomanNumeralCaps{1}. $2)$: if $s=i$, the packet starts in the
  $i,j$ contraction. Since $a=\pi'_s(\bot)=\pi_{i,j}(\bot)=\pi_i(\bot)$
  (Case \RomanNumeralCaps{2}),
  we deduce that $P$ describes possibly some hops between $i$ and $j$, and
  then $q$. Since $(i,j)$ transitions are rewritten $(i,i)$ in $Q$ and
  removed in $R$, the first node that is not $i$ nor $j$ to appear in
  $P$ must be $a$. Thus $q=a$.

\emph{Induction step:} Assume the following statement holds for any
$k'\leq k$: $Q[k']=P'[k']$. We prove that necessarily $Q[k]=P'[k]$.
Let $Q[k-1]=(v_-,v),P'[k-1]=(v_-,v),Q[k]=(v,q), P'[k]=(v,a)$. We need to
show that $a=q$. We again start by the simplest case $1)$: $v\not\in V'(i)
\cup \{i\}$: since $\pi_v=\pi'_v$ by Case \RomanNumeralCaps{1},
$q=\pi_v(v_-)=\pi'_v(v_-)=a$.
$2)$: if $v\in V'(i)$. We need to look at $v_-$. If $2.1)$: $v_-\neq i$ we
again directly use Case \RomanNumeralCaps{1}:
$\pi_v(v_-)=\pi'_v(v_-)$. If $2.2)$: $v_-= i$ the packet just left the
$(i,j)$ contraction. In $P$, $v_-$ can correspond to either $i$ or $j$
in $P$. If it corresponds to $i$, use Case \RomanNumeralCaps{1}. If it
was a $j$, since in Case \RomanNumeralCaps{3} we replaced $j$'s
connections (in $A$) by $i'$s connections (in $A'$), and since no node
$v$ is both a neighbor of $i$ and $j$ in $F\cup R$, this mapping
uniquely applies.
$3)$: if $v=i$: the packet is in the $(i,j)$
contraction. By definition of $Q$, necessarily $v_-\neq i$ or
$j$. Since $\pi'_i(v_-)=\pi_{ij}(v_-)$ we deduce $q=a$.
\end{proof}

\begin{proof}[Proof of \cref{cont2}]
	Let $R=\{(j,r),r\in V(i)\cap V(j)\}$. 
	Let $A'$ be the $(i,j)$ contracted algorithm of $A$.
  Let $F$ be a set of link failures of $G'$. 
	Observe that if $s,t$ are connected in $G'\setminus F$
	then $s,t$ are connected in $G\setminus F\cup R$.
  Let $s,t$ be connected in $G'\setminus F$.
	Let $p_{A'}$ (resp. $p_{A}$) be the path produced by $A'$ on $G'\setminus F$ (resp.\ $A$ on $G\setminus (F \cup R)$) from $s$ to $t$.

  Since $A\in A_p$ and $s,t$ are connected in $G\setminus (F\cup R)$,
  then necessarily $p_A$ is finite and ends up in $t$. Therefore,
  \cref{cont1} implies $p_{A'}$ is finite and ends in $t$
  too: $A'$ succeeds.
	%
\end{proof}

\section{Proof of Theorem~\ref{thm:feigenbaum-gadget}} \label{sec:feigenbaum-proof}

\begin{proof}[Proof of Theorem~\ref{thm:feigenbaum-gadget}]
  Construct the gadget $G$ as follows. Take four nodes 1, 2, 3, and 4 and connect them by a link to a center node $c$. Then, for each unordered pair $(i,j)$ of nodes from $\{1,2,3,4\}$ create a new node $ij$, and connect $ij$ to $i$ and $j$. Connect all $ij$ to a new target node $t$. Finally, create a source node $s$ and connect it to $\{1,2,3,4\}$. $G$ is illustrated in Figure~\ref{fig:feigenbaum-gadget}a.

  We will call nodes 1, 2, 3, and 4 \emph{level one} nodes, and nodes 12, 13, 14, 23, 24, and 34 \emph{level two} nodes.

  We first claim that if nodes of levels one and two have degree 2 after failures, then they must \emph{always} forward the packet coming from $s$, level one, or level two to the other port, with the exception of forwarding back towards $s$.
  \begin{itemize}
    \item If node $vu$ on level two is connected to $t$, we can cut all other links to $t$ and force $vu$ to forward to $t$, as otherwise the forwarding fails.
    \item If $uv$ is not connected to $t$, we can again create a unique path to $t$ that goes through $uv$. Choose some $vw$ on level two and cut all links to $t$ except $\{vw,t\}$. In addition cut all links to $vw$ except $\{v,vw\}$, and cut all links to $v$ except $\{uv, v\}$. A packet coming from $\{u,uv\}$ must be forwarded to $\{uv,v\}$. The other direction is symmetric.
    \item If a level one node $v$ is connected to $s$, then by setting $F = \{ \{s,u\} : u \neq v \}$ any non-forwarding rule at $v$ fails.
    \item If a level one node $v$ is not connected to $s$, then it is connected to some level two nodes $uv$ and $wv$ or $c$ and some level two node $uv$. In the first case, by cutting all links from level two to $t$ except for $uv$ or $wv$, and, respectively, the link $\{u, uv\}$ or the link $\{w,wv\}$, we force the unique path from $s$ to $t$ to take the links $\{vu, v\}$ and $\{wv, v\}$, in both directions. In the second case, we can create the following unique path from $s$ to $t$ by failing all other links: $(s, u, uv, v, c, w, wu, t)$, where $w$ is some node not $v$ or $u$. This forces forwarding at $u$ from $uw$ to $c$. For the other direction fail everything except the path $(s, u, c, v, vw, t)$ for some $w$.
  \end{itemize}
  Any perfectly resilient forwarding pattern must therefore forward packets over degree-2 nodes.

  We assume that there are no failures around $c$ so the forwarding function $\pi_c^{s,t}$ remains fixed. In the following we identify the incident links of $c$ with the neighbors they connect to. If $\pi_c^{s,t}(i) = i$ for any $i \in [4]$ the forwarding fails when $s$ is connected only to $i$ and $i$ only to $c$. Assume without loss of generality that $\pi_c^{s,t}(1) = 2$. By having $\pi_c^{s,t}(2) = 1$ we have a loop when all other links incident to 2 fail and 1 has links $\{s,1\}$ and $\{1,c\}$.  Without loss of generality assume that $\pi_c^{s,t}(2) = 3$. By a similar argument we see that we must have $\pi_c^{s,t}(3) = 4$, as otherwise we would create a loop. Now consider different settings for $\pi_c^{s,t}(4)$.
  \begin{itemize}
    \item $\pi_c^{s,t}(4) = 1$: we fail all links incident to level one except $\{s,4\}, \{1,13\}, \{13,3\}$, all links incident to $c$ or 2. The forwarding pattern loops $(s, 4, c, 1, 13, 3, c, 4)$ and then either visits $s$ or goes back to $c$ immediately, finishing the loop.
    \item $\pi_c^{s,t}(4) = 2$ or $\pi_c^{s,t}(4) = 3$: These cases are similar, as we can force the packet to never visit 1, which could be the unique remaining path to $t$. Fail all links of $s$ except $\{s,4\}$, and all other links of 4 except $\{4,c\}$. In addition fail all links from 2 and 3 to level two. The packet will start a loop on $(s, 4, c, 2, c, 3, c, 4)$ or on $(s, 4, c, 3, c, 4)$, respectively. Then it will either visit $s$ or go directly back to $c$, finishing the loop.
  \end{itemize}
  As all choices for $\pi_c^{s,t}(4)$ create a forwarding pattern with a loop, a perfectly resilient forwarding pattern cannot exist.
\end{proof}
\section{Proof of Theorem~\ref{thm:planar-s}} \label{sec:planar-s}

\begin{proof}
We note that by choosing the node 5 in Figure~\ref{fig:planar} as the source, the proof for Theorem~\ref{thm:planar} would directly carry over, if we could apply Lemma~\ref{lemma:orbit} to all nodes in Figure~\ref{fig:planar}.
However, Lemma~\ref{fig:planar} only considers target-based routing that does not consider the source, and we hence need to utilize \cref{lem:s-orbit} for node 1 (all relevant neighbors are in the same orbit) and \cref{corr:s-paths-deg} for the other nodes, where the degree two case suffices.
In order to apply Lemma~\ref{lem:s-orbit} to node 1, we need to connect the new source $s$ to two relevant neighbors of $s$, where we pick nodes 3 and 5, which keeps the 7+1 node construction planar.
Note that for the actual packet path in Figure~\ref{fig:planar}(a,b,c), we can fail the link $(s,3)$, keeping $(s,5)$ alive, and if $5$ were to route back to $s$, $s$ can only bounce the packet back.

It remains to force all degree 2 nodes in Figure~\ref{fig:planar}(a,b,c) to relay incoming packets through the other port, via Corollary~\ref{corr:s-paths-deg}, which we prove by case distinction.
To this end, for each of the seven cases (nodes 2,3 for a, 2,4 for b, and 3,4,6 for c), we need to show that there are pairs of node-disjoint paths for each such node $i$, with two different surviving links $(v_1,i),(v_2,i)$, that are in $E(G)$ minus all links incident to $i$, except $(v_1,i),(v_2,i)$:\footnote{Recall that $i$ is only aware of its incident link failures.}
\begin{itemize}
	\item start from $s$ and end with $(v_1,i)$, and start with $(v_2,i)$ and end at $t$,
	\item start from $s$ and end with $(v_2,i)$, and start with $(v_1,i)$ and end at $t$.
\end{itemize}

\begin{itemize}
	\item Figure~\ref{fig:planar}(a): Node 2 $\left[(3,2),(2,1) \right]$
	\begin{itemize}
		\item $s-5-6-3-2$ \textbf{:} $2-4-t$ and $s-5-1-2$ \textbf{:} $2-3-t$
	\end{itemize}
	\item Figure~\ref{fig:planar}(a): Node 3 $\left[(1,3),(3,2) \right]$
	\begin{itemize}
		\item $s-5-1-3$ \textbf{:} $3-2-t$ and $s-5-1-2-3$ \textbf{:} $3-6-t$
	\end{itemize}
	\item Figure~\ref{fig:planar}(b): Node 2 $\left[(1,2),(2,4) \right]$
	\begin{itemize}
		\item $s-5-1-2$ \textbf{:} $2-4-t$ and $s-5-1-4-2$ \textbf{:} $2-1-t$
	\end{itemize}
	\item Figure~\ref{fig:planar}(b): Node 4 $\left[(2,4),(4,1) \right]$
	\begin{itemize}
		\item $s-5-1-4$ \textbf{:} $4-2-3-t$ and $s-3-2-4$\footnote{Note that the packet is not actually routed via $3$ in Figure~\ref{fig:planar}(b), but the node 4 must also provision for this case in order to guarantee perfect resilience.} \textbf{:} $4-1-5-6-t$
	\end{itemize}
	\item Figure~\ref{fig:planar}(c): Node 3 $\left[(1,3),(3,6) \right]$
	\begin{itemize}
		\item $s-5-1-3$ \textbf{:} $3-6-t$ and $s-5-6-3$ \textbf{:} $3-1-4-t$
	\end{itemize}
	\item Figure~\ref{fig:planar}(c): Node 4 $\left[(1,4),(4,6) \right]$
	\begin{itemize}
		\item $s-5-1-4$ \textbf{:} $4-6-t$ and $s-5-6-4$ \textbf{:} $4-1-3-t$
	\end{itemize}
	\item Figure~\ref{fig:planar}(c): Node 6 $\left[(3,6),(6,4) \right]$
	\begin{itemize}
		\item $s-5-1-3-6$ \textbf{:} $6-4-t$ and $s-5-1-4-6$ \textbf{:} $6-3-t$
	\end{itemize}
\end{itemize}
\end{proof}

\section{Proof of \cref{thm:2hops}} \label{sec:2hops}
\begin{proof}
Resilience can be ensured with the following forwarding pattern.
Let $V(s)$ be the set of neighbors of the source node $s$
which are on a path of length two to the target $t$ before
the failures.
We define the forwarding function $\pi^{s,t}_s$ of the source $s$ as follows:

\begin{itemize}
\item We order the neighbors $V(s)$ of $s$ arbitrarily, i.e.,
$V(s)=(v_1,v_2,\ldots,v_k)$, and source $s$ first tries to forward
to the first neighbor $v_j$ which is connected to it after the failures ($\pi_s^{s,t}(\bot) = v_j$), i.e., in a skipping fashion.

\item When a packet arrives on the in-port from $v_i$, for any $i$,
$s$ forwards the packet to $v_{j}$, the next neighbor in $V(s)$ such that $\{s, v_j\} \notin F$, i.e., skipping as well.
\end{itemize}

For all other nodes $v$, the forwarding function is defined as follows.
When a packet arrives from the source $s$, send it to the target if the link has not failed. Otherwise return the packet to the source.
Note that this is a skipping pattern, as we can predefine a skipping of $(v,t)$ to route to $s$.
This ensures that the source $s$ tries all of its neighbors which could be
on a path of length two to the target. Since we assumed such a path existed, the forwarding pattern must succeed in routing the packet to the target.
\end{proof}

\section{Proof of \cref{thm:2hops-no-s}} \label{sec:2hops-no-s}
\begin{proof}
Let us assign a unique ID to every node and construct the forwarding pattern based on these identifiers.
The forwarding function of each node $v \in V \setminus \{s,t\}$ is defined as follows:
\begin{itemize}
	\item If $t$ is a neighbor of $v$, forward to $t$.
	\item If $t$ is not a neighbor, then
	\begin{itemize}
		\item if $v$ does not have the lowest ID in its neighborhood, forward to the neighbor with lowest ID, i.e., skipping potential lost neighbors with lower ID,
		\item if $v$ has the lowest identifier in its neighborhood, then route according to a skipping pre-defined cyclic permutation of all neighbors.
	\end{itemize}
\end{itemize}
Observe that, unless the target is reached on the way, the packet will reach a node $w$ that has the lowest identifier in its two-hop neighborhood.
From there on, the argument is analogous to the proof of Theorem~\ref{thm:2hops}:
The node $w$ will forward the packet to all its one-hop neighbors, which in turn bounce it back to $w$, unless connected to $t$. Since, by assumption, one of them is connected to $t$, the forwarding succeeds.
\end{proof}

\section{Proof of \cref{thm:outerplanar}} \label{sec:outerplanar}

\begin{proof}
Fix an arbitrary outerplanar embedding of the graph without failures (a consistent 
embedding must be used on all nodes when constructing the forwarding pattern). 
Without loss of  generality we define the canonical direction to be clockwise and we 
number the ports of the nodes accordingly starting from an arbitrary port. 

Note that if a link $l$ belongs to the outer face of a planar graph $G$, it also belongs
to the outer face for all subgraphs of $G$, in particular also in $G\setminus F$ for $l \notin F$. Since 
all nodes belong to the outer face for an outerplanar graph, it is hence enough to 
demonstrate that the forwarding pattern ensures packets use only the links of the 
outer face and do not change the direction despite failures.

Given an in-port belonging to a link on the outer face of a graph in clockwise 
direction we can decide locally which out-port belongs to the outer face to continue 
the walk using the right-hand rule encoded in the the forwarding function with $F$.

The routing algorithm assigns the following forwarding functions to node $v$ with 
degree $d$: $\pi_v^t(x) = x +1 \bmod d$. Hence the forwarding
function assigns out-port~$x +1 \bmod d$  to in-port $x$.
When node $s$ sends a packet to $t$ it picks an arbitrary out-port
that belongs to the outer face in the clockwise direction.
When failures occur, the skipping method is used to derive the next usable
output. 

As we will see next, this is enough to encode the information of which links belong 
to the outer face of $G\setminus F$. To this end, we show that this forwarding pattern 
routes a packet along the links of the outer face for an arbitrary link failures $F$ in 
clockwise direction.

We proceed to show this claim for the source and then for all other nodes on the
path to the target by induction.

Each node has at least one pair of an incoming port and an outgoing
port on the outerplanar face of $G$ (for a node with degree one these two
ports belong to the same link).  Let $(i_j, o_j)$ denote the $j^{th}$
such pair out of $k$ for node $v$ in $G$, starting to number them from
port 0 at node $v$ in clockwise order.
Note that it holds for all these $k$ pairs that $o_j = i_j + 1 \bmod d$.

Without failures, the source selects an outgoing port on the outer face in
the correct direction by definition. Without loss of generality, let this out-port be
$o_0$ from the pair $(i_0, o_0)$ defined above. If the corresponding link
has failed, the routing scheme will try port $o_0 + j \bmod d$ for
$j = 1, \dots, d-1$ until successful according to the skipping forwarding pattern.
Due to the failures, the corresponding link is now on the outer face and the
statement thus holds for the base case.

For a subsequent node $v$ on the walk to the target the packet will enter the
node on  port $i_j$  for some $j$ if no failures affect this node. In this case,
the routing scheme selects $i_j + 1 \bmod d$ as the out-port, which
happens to be $o_j$ as discussed above. Thus the packet remains using only links
of the outer face travelling in the right direction on $G\setminus F$. If there has been a failure affecting
node $v$, then the outgoing port might no longer be available. As shown for the source,
iterating over the next ports clockwise will guarantee that the packet is still using a link
of the outer face of $G\setminus F$. If the link incident to $i_j$ is down, then $v$'s neighbors
will select another link and the packet might enter the node on a port that
is not among the pairs for the graph without failures. Nevertheless, due to
the arguments above this implies that the incoming port now is part of
the walk along the links of the outer face in clockwise direction
 and also the outgoing port
chosen subsequently is part of the links on the outer face of $G\setminus F$.
Hence the routing scheme succeeds to lead packets to the target successfully
 without maintaining state in the packet or at the nodes.
\end{proof}

\section{Proof of \cref{cor:sameface}} \label{sec:sameface}

\begin{proof}
The algorithm from the proof for Theorem~\ref{thm:outerplanar} only needs to be
adapted at the start. Since the source and the target might not be
part of the outer face in $G$, the source chooses an arbitrary (clockwise direction) out-port that belongs to the face it shares with the target (some nodes may share several faces with the target, any of them can be chosen in this case).
Since failures only reduce the number of faces, the shared face may contain more
links after failures but not fewer and both source and target will always
belong to the same face on $G\setminus F$, as long as the graph is still connected.
Thus following the links of the face in clockwise direction will lead to the target
for arbitrary non-disconnecting failures.
\end{proof}
Another way of thinking about the proof is that we can also choose the embedding of the planar graph s.t.\ any desired face is the outer face, see e.g.\ Schnyder~\cite{DBLP:conf/soda/Schnyder90}.
%
As before, note that our forwarding pattern can visit every node on the face despite failures. 

\section{Proof of \cref{thm:k4}} \label{sec:thm:k4}
\begin{proof}
Let the four nodes of the $K_4$ w.l.o.g.\ be $\{v_1,v_2,v_3,v_4=t\}$.
The node $v_4=t$ must have at least one connected neighbor, w.l.o.g.\ $v_3$, where we are done if we start on $v_3$ or $v_4$.
It remains to fix the routing for any node not neighboring the target (where we directly reach $t$), for which we choose a cyclic permutation of all neighbors, picking skipping.
Hence, if the packet starts on w.l.o.g.\ $v_1$, where $v_1$ is part of the connected component of $t$, and reaches $v_3$ or $v_4$, perfect resilience is achieved.
We now perform a case distinction, where we note that while the node $t$ is distinguishable as a neighbor, for $v_1$ it is unknown which of its neighbors is connected to $t$.
Assume $v_1$ has only one neighbor, then forwarding succeeds after two or three hops (a line).
If $v_1$ has two neighbors and picks the one neighboring the target, we are done after two hops.
Else, assume it picks the wrong one, say $v_2$.
Then, if $v_2$ only has $v_1$ as a neighbor, the packet bounces back, then reaches $v_3$ and then $t$.
Lastly, if $v_2$ is connected to $t$, we are done, and if not, it is connected to $v_3$, which then forwards the packet to $t$.
\end{proof}

\fi
\ifFull
\fi

\end{document}